\documentclass{article}
\usepackage{arxiv}
\renewcommand{\shorttitle}{Scalable Composition of Byzantine Agreements under Reorder Attacks}

\usepackage[utf8]{inputenc}
\usepackage[T1]{fontenc}
\usepackage{amsmath,amsfonts,amssymb,amsthm}
\usepackage{amsopn}
\usepackage{algorithm}
\usepackage{algorithmic}
\usepackage{array}
\usepackage{tabularx}
\usepackage[caption=false,font=normalsize,labelfont=sf,textfont=sf]{subfig}
\usepackage{textcomp}
\usepackage{url}
\usepackage{verbatim}
\usepackage{graphicx}
\usepackage{xcolor}
\usepackage[hidelinks]{hyperref}
\hypersetup{
  pdftitle={Scalable Composition of Byzantine Agreements under Reorder Attacks},
  pdfauthor={J. Chen, J. Dong, J. Li, X. Xia, and W. Zhou}
}
\def\BibTeX{{\rm B\kern-.05em{\sc i\kern-.025em b}\kern-.08em
    T\kern-.1667em\lower.7ex\hbox{E}\kern-.125emX}}

\floatname{algorithm}{Protocol}

\theoremstyle{plain}
\newtheorem{theorem}{Theorem}
\newtheorem{lemma}[theorem]{Lemma}
\newtheorem{definition}[theorem]{Definition}
\newtheorem{corollary}[theorem]{Corollary}

\begin{document}
\title{Scalable Composition of Byzantine Agreements under Reorder Attacks\thanks{A preliminary version of this work appeared in the Proceedings of the 7th international conference on Advances in Financial Technologies (AFT'25). The authors thank several anonymous reviewers for their valuable comments for the conference version of this work, and Guillermo Angeris and Gerui Wang for helpful discussions. This work is partially supported by Beijing Advanced Innovation Center for Future Blockchain and Privacy Computing.}}

\author{%
  Jing Chen\thanks{J. Chen is the corresponding author. Department of Computer Science and Technology, Tsinghua University, Beijing, China (jchencs@tsinghua.edu.cn).}
  \And
  Jin Dong\thanks{J. Dong, Beijing Academy of Blockchain and Edge Computing (BABEC), Beijing, China (dongjin@baec.org.cn).}
  \And
  Jichen Li\thanks{J. Li, Department of Computer Science and Technology, Tsinghua University, Beijing, China (jichenli@mail.tsinghua.edu.cn).}
  \And
  Xuanzhi Xia\thanks{X. Xia, Department of Computer Science and Technology, Tsinghua University, Beijing, China (xiaxz24@mails.tsinghua.edu.cn).}
  \And
  Wentao Zhou\thanks{W. Zhou, Department of Computer Science and Technology, Tsinghua University, Beijing, China (zhouwt24@mails.tsinghua.edu.cn).}
}

\date{September 9, 2026}

\setcounter{footnote}{-1}

\maketitle

\begin{abstract}
  Byzantine agreement (BA) is a foundational building block in distributed systems and has been extensively studied for decades. With the growing demand for protocol composition in practice, the security analysis of BA protocols under multi-instance executions has attracted increasing attention. However, most existing adversary models focus solely on party corruption and neglect important threats posed by adversarial manipulations of communication channels in the network. Through channel attacks, messages can be reordered across multiple executions and lead to violations of the protocol's security guarantees, without the participating parties being corrupted.
  
  In this work, we present the first adversary model that combines party corruption and channel attacks. Based on this model, we establish new security thresholds for Byzantine agreement under parallel and concurrent compositions, supported by complementary impossibility and possibility results that match each other to form a tight bound. For the impossibility result, we show that even authenticated Byzantine agreement protocols cannot be secure under parallel composition when $n \leq 3t$ or $n \leq 2c + 2t + 1$, where $t$ and $c$ denote the number of corrupted parties and communication channels, respectively, and $n$ is the number of parties.

  For the possibility result, we prove the existence of secure protocols for unauthenticated Byzantine agreement under parallel and concurrent composition, when $n > \max\{3t, 2c+2t+1\}$.
  We first provide general black-box compilers that transform any single-instance secure BA protocol into one that is secure under parallel and concurrent executions without additional security assumptions. To optimize performance, we further design refined compilers using erasure-correcting codes. These refined versions significantly reduce communication overhead, particularly for long messages, where they achieve a constant multiplicative overhead compared with the original protocol, thus achieving the same asymptotic communication complexity.

\end{abstract}


\section{Introduction}
\label{sec:intro}


\begin{table*}[!t]
\renewcommand{\arraystretch}{1.25}
\caption{Communication Complexity of our Concurrently Secure Primitives}
\label{tab:complexity_results}
\centering
{\footnotesize
\begin{tabularx}{\textwidth}{|l|>{\centering\arraybackslash}X|>{\centering\arraybackslash}X|>{\centering\arraybackslash}X|}
\hline
\textbf{Primitive} & \textbf{Short Messages} & \textbf{Long Messages$^\ddagger$} & \textbf{Very Long Message$^\ast$} \\
\hline
\textbf{RMT} & $O(n|M| + n\log n)$ & $O(|M| + \lambda n \log n)$ & $O(|M| + \lambda n \log n)$ \\
\hline
\textbf{RB} & $O(n^3 |M| + n^3\log n)$ & $O(n^2|M| + \lambda n^3 \log n)$ & $O(n|M| + \lambda n^3 \log n)$ \\
\hline
\end{tabularx}}
\vspace{0.5ex}
\begin{minipage}{\textwidth}
\footnotesize
\setlength{\parskip}{0.15em}
$^\ddagger$ Long messages: $|M| \in \Omega(\lambda \log n)$. \\
$^\ast$ Very long messages: $|M| \in \Omega(\lambda n \log n)$. \\
Protocols for long and very long messages are secure against a polynomial-time adversary. Protocols for short messages remain our general compilers and secure against an unbounded adversary.
\end{minipage}
\end{table*}

\begin{table*}[!t]
\renewcommand{\arraystretch}{1.2}
\caption{Communication Complexity of Original and Compiled Asynchronous Consensus Protocols}
\label{table:before-after}
\centering
{\footnotesize
\setlength{\tabcolsep}{3pt}
\begin{tabularx}{\textwidth}{|l|>{\hsize=1.1\hsize\raggedright\arraybackslash}X|>{\hsize=0.967\hsize\centering\arraybackslash}X|>{\hsize=0.967\hsize\centering\arraybackslash}X|>{\hsize=0.956\hsize\centering\arraybackslash}X|}
\hline
Protocol & Input Length $L$ & Assumptions & Orig.\ Comm. & Compiled Comm. \\
\hline
Huang et al.~\cite{huang2024byzantine} & $|L| \in \Omega(1)$ & None & $\tilde{O}(n^{12}|L|)$ & $\tilde{O}(n^{13}|L|)$ \\
\hline
Das et al.~\cite{das2024asynchronous} & $|L| \in O(\lambda \log n)$ & Random Oracle & $O(\lambda n^3)$ & $O(\lambda n^4)$ \\
& $|L| \in O(\lambda n \log n)$ &   & $O(\lambda n^3 \log n)$  & $O(\lambda n^4 \log n)$ \\
& $|L| \in \Omega(\lambda n^2 \log n)$ &   & $O(n^2 |L|)$ & $O(n^2 |L|)$ \\
\hline
Guo et al.~\cite{guo2022speeding} & $|L| \in O(\lambda n \log n)$ & Threshold Public Key Encryption & $O(\lambda n^3 \log n)$ & $O(\lambda n^4 \log n)$ \\
 & $|L| \in \Omega(\lambda n^2 \log n)$ &  & $O(n^2 |L|)$ & $O(n^2 |L|)$ \\
\hline
\end{tabularx}}
\vspace{0.5ex}
\begin{minipage}{\textwidth}
\footnotesize
\setlength{\parskip}{0.15em}
The notation $\tilde{O}(\cdot)$ in the original paper omits poly-logarithmic factors in $n$. \\
Based on the original protocols' message lengths and by properly choosing our primitives, we achieve concurrent security without adding any extra assumption.
\end{minipage}
\end{table*}

The Byzantine agreement (BA) problem~\cite{pease1980reaching} is a foundational challenge in distributed systems, concerned with achieving consensus among parties even in the presence of failures or malicious behavior.
This problem has been extensively studied under various models, leading to significant theoretical breakthroughs, including both impossibility results \cite{lamport1982the,fischer1985impossibility} and the development of innovative protocols \cite{feldman1997opt, castro1999practical,pfitzmannW1992unconditional,katzK2009expected}.
As a core mechanism for achieving distributed consistency and fault tolerance, Byzantine agreement underpins a wide range of applications, including blockchain technologies~\cite{chen19algorand,malek2005fault,qu2023quantum,yoo2019formal,kuo2020fair,zhao2023secure}, secure multiparty computation~\cite{yao82protocols,goldreich87how,michael1988completeness,chaum1988multiparty,deligios2021round,fitzi2002unconditional,gennaro2002on} and diverse distributed services~\cite{galil1987cryptographic,locher2020fast,agrawal2016performance}.

Much of the early research on the Byzantine agreement problem was conducted under a classical model where a network of $n$ parties, each holding an initial input value, communicates over reliable synchronous channels. 
In this model, faults are restricted to the parties themselves: up to $t$ parties may be corrupted by an adversary, but the communication network is assumed to be secure ---messages cannot be forged, altered, or dropped by the adversary once sent by an honest party.
The standard goals of a BA protocol in this setting are:
(1) {\em agreement} ---all honest parties eventually terminate and output the same value; and
(2) {\em validity} ---if all honest parties start with the same input, that value is the output. 
Studies show that in unauthenticated settings, BA is achievable if and only if $t < \frac{n}{3}$~\cite{pease1980reaching,lamport1982the}, while in authenticated settings, where digital signatures prevent forgery~\cite{rivest1978a,goldwasser1988digital}, BA can be achieved with any number of corruption $t < \frac{n}{2}$~\cite{fitzi2003generalized}. 
The expected constant-round protocols with optimal resiliency in the two settings were then constructed in \cite{feldman1988optimal} and \cite{katzK2009expected}, respectively.

However, the aforementioned studies primarily focus on single-instance executions of Byzantine protocols, leaving the critical issue of protocol composition unaddressed. In modern distributed environments, such as sharded blockchains~\cite{zamani2018rapidchain,zhang2020cycledger,zhang2025scaling}, multiparty computation frameworks~\cite{saha2024application,gai2025scheme}, or cross-chain settings~\cite{yin2021sidechains,chen2026efficient}, multiple consensus protocols may be composed sequentially, in parallel, or concurrently.
In sequential composition, each protocol instance begins only after the previous one has been completed. In parallel composition, all instances are initiated at the same time and proceed with their steps aligned with each other. The most general model, concurrent composition, grants the adversary full control over the start times and execution rates of different instances.

The shift from single-instance to multi-instance execution introduces new problems, demanding a re-examination of existing assumptions about fault models and protocol security. In particular, \cite{lindell2006composition} showed that, without a unique and common session identifier for every execution of the protocol, no authenticated BA protocol can remain secure even under just two parallel executions when corruption exceeds $n/3$. In other words, authenticated BA performs as poorly as unauthenticated BA under parallel composition.
Common session IDs can sometimes be achieved via a bootstrap phase, but in this work we would like to pursue the power of stateless composition without relying on such a bootstrap, so that the resulting BA protocols can be directly applied whenever compositional executions are needed.

The challenge with protocol composition is that 
the adversary’s power is amplified, and in our study we will amplify it even further. Indeed, a key limitation of the classical model is its implicit assumption that adversarial power is confined to corrupting parties, while communication channels remain trustworthy in the sense that messages sent between honest parties will be received correctly. This abstraction overlooks more realistic adversarial behaviors, such as man-in-the-middle or message-reordering attacks, which target the network rather than the parties themselves, 
especially when protocol composition is concerned.
By carefully disrupting or manipulating some communication channels between multiple protocol instances, an adversary may induce cross-protocol interference, leading to violations of agreement or validity that would not occur in isolated executions. 

\subsection{Adversary Model}
To address the limitations of classical adversary models, we introduce an extended adversarial framework tailored to protocol composition environments. 
In this model, the adversary retains the classical ability to corrupt up to $t$ parties, allowing them to deviate arbitrarily from the protocol. 
In addition, channels can also be corrupted, without sending and receiving parties being corrupted or even aware of the channel attack. 
This reflects realistic attack approaches in networks where multiple protocols share the underlying network infrastructure.

In the real world, the adversary may exploit man-in-the-middle attacks, such as ALPACA~\cite{ALPACA}, to redirect the messages and bypass the security guarantees of TLS and application layer protocols, resulting in cross-protocol attacks that rearrange the messages from different protocols without breaking the signature scheme.

Such a setting captures the inherent interdependence of parallel or concurrently running consensus protocols, where messages from different instances may traverse overlapping physical or logical channels. 
Through reorder attacks, 
the adversary can introduce inter-instance inconsistencies that undermine global agreement guarantees, even if each protocol remains individually secure under classical assumptions.

This hybrid adversary model, which combines party corruption with channel attacks, introduces new theoretical challenges. 
In particular, these two kinds of attacks are coupled together, thus determining the exact resilience thresholds under simultaneous control of $t$ parties and $c$ channels requires a fresh analytical approach that goes beyond traditional BA models. We formally define the adversary model in Section \ref{sec:model}.

\subsection{Our Results}

Our main results establish a new tight security threshold under the new adversary model.
We first present the {\em impossibility} result (Theorem~\ref{thm:impossibility}) in Section \ref{sec:impossibility} showing that even in the authenticated setting, BA under parallel (and concurrent) composition is not achievable if $n\leq 3t$ or $n\leq 2c+2t+1$, where $t$ is the number of corrupted parties and $c$ is the number of channels that the adversary can manipulate.

Intuitively, even if a majority of parties are honest, agreement may still fail if the adversary can confuse one party, making it unable to tell which specific protocol instance it is participating in. For example, when two protocols are running in parallel, an attacker can target a particular party and swap the messages it is supposed to receive in the two protocols, causing its result to differ from that of other honest parties. Therefore, it is not enough to simply bound the number of corrupted parties ---one must also ensure that the honest parties maintain a sufficient number of communication with parties in the same protocol.
This leads to a trade-off between party corruption and channel interference, and motivates the combined threshold $n>2c+2t+1$ as the condition necessary to preserve the communication advantage of the honest majority.

We further prove {\em possibility} results in Sections \ref{sec:black-box} and \ref{sec:concurrentBA} under compositional executions, showing that the condition $n>2c+2t+1$ along with the condition $n>3t$ in classical setting (which is the tight bound when only party corruption is considered) is sufficient even in the unauthenticated setting. 
That is, when the combined adversarial power satisfies both conditions, it is possible to design a BA protocol that remains secure under arbitrary parallel and concurrent executions.  
These results complement our impossibility result, forming a tight security threshold: $n>\max\{2c+2t+1, 3t\}$ is both necessary and sufficient for protocol composition in our adversary model.
This tight threshold determines \emph{when} Byzantine agreement is achievable under composition, and we further fill the gap regarding \emph{how} to convert any authenticated protocol into one that is secure under concurrent and parallel composition.

Noticeably, in Section \ref{sec:black-box} we provide a general black-box compiler that transforms any single-instance secure BA protocol into one that remains secure under {\em parallel} executions, in synchronous networks (Theorem~\ref{thm:parallel}).
Extending this result, in Section~\ref{sec:concurrentBA} we propose a general black-box compiler designed for the more challenging setting of {\em concurrent} composition in asynchronous networks. 
The core of this transition lies in the construction of two fundamental communication primitives ---Reliable Broadcast (RB) and Reliable Message Transmission (RMT)---  in a concurrent environment (Theorems~\ref{thm:concurrent-RB} and \ref{thm:concurrent-RMT}).

While our general black-box compilers
are assumption-free and apply to an unbounded adversary, 
they incur an $O(n)$ or $O(n\log n)$ multiplicative factor in the resulting protocol's communication complexity, depending on the protocol's input length. 
To reduce the complexity in practice, in Section~\ref{sec:long} we design highly efficient RMT and RB primitives optimized for long messages (Theorems~\ref{thm:long-RMT} and~\ref{thm:long-RB}). 
We leverage erasure-correcting codes and cryptographic proof schemes, thus shifting from an unbounded adversary to a polynomial-time adversary when such constructions are adopted. 
By doing so, we reduce the complexity dependence of both primitives on the message length $|M|$, 
at the tradeoff of an additive factor now also depending on the security parameter~$\lambda$.

In Section \ref{sec:long} we formalize our complexity results for concurrent protocol executions in asynchronous networks. The same technique naturally applies to parallel executions and we omit this part of the discussion from the paper for succinctness.
As summarizations, Table~\ref{tab:complexity_results} provides the communication complexity of our refined primitives. Also, Table~\ref{table:before-after} summarizes the overall complexity of our compiled protocols compared with the original protocols, when our technique is applied to several representative asynchronous consensus protocols to make them concurrently secure. Here $|L|$ represents the input length of the protocol. Noticeably, those original protocols currently achieve the best communication efficiency under respective security assumptions. 

In conclusion, our work provides the first characterization of security for BAs under protocol composition with both party corruption and channel attacks, filling a critical gap in the theoretical understanding of consensus under adversarial network interference. 
By introducing a unified adversary model and establishing the tight security threshold, 
we extend the classical BA framework to more accurately reflect the complication of modern distributed systems.

\section{Related Work}

Research on Byzantine agreements has a long tradition,
and with the development of blockchains there has been a growing body of literature that recognizes the importance of this area.

\textbf{Composition of Byzantine Agreements.}
In a stateless model where protocols do not have distinguishing IDs, which is also considered by our work, \cite{lindell2006composition} studied the sequential, parallel, and concurrent composition of authenticated BA protocols and presented the impossibility result as mentioned in the introduction. It also constructed secure randomized protocols under sequential composition.
Without cryptographic primitives, \cite{benor2003Resilient} considered concurrently secure BAs with expected-constant round in both synchronous and asynchronous networks with optimal resiliency in the classical model (i.e., $n>3t$). 
However, as pointed out by \cite{cohen2023concurrent}, its security proof had some subtle issues regarding, e.g., the use of oblivious leader election.

Neither \cite{benor2003Resilient} nor \cite{cohen2023concurrent} is stateless: they assume a unique and common session ID for every execution of the protocol. 
As pointed out by \cite{lindell2006composition}, concurrent composition of any number of executions in this case is possible. Indeed, the focus of \cite{benor2003Resilient} and \cite{cohen2023concurrent} was to construct expected-constant round protocols that are concurrently composable. 
Similarly, \cite{Gupta2012ANL} assumed unique session IDs but considered an adversary who can corrupt a player in some but not all parallel executions.

Moreover, \cite{cohen2023concurrent}  followed the Universal Composability (UC) framework~\cite{Canetti2001uc}, which means the protocol is secure under arbitrary composition with itself and other cryptographic protocols
in an adversarially controlled manner. 
Focusing on the asynchronous setting, \cite{cohen2023concurrent} considered concurrent BA and presented the first information-theoretic multi-valued oblivious common coin (OCC) protocol with optimal resiliency. It further provided a modularized protocol for round-preserving parallel composition of BA that was simpler than the construction in \cite{benor2003Resilient}.

\smallskip
\textbf{Computation and Communication Complexity.}
An important line of work considers the complexity of a single execution of 
a Byzantine agreement. A body of foundational studies, in particular~\cite{dolev1983authenticated, dwork1988consensus, castro1999practical}, proposed Byzantine agreement protocols that required a quadratic number of messages. 
More specifically, the $O(n^2)$ communication complexity necessitates nearly all-to-all communication between parties. Subsequently, for {\em randomized Byzantine protocols}, a seminal work~\cite{king2006scalable} proposed a protocol where each party speaks to only $\Tilde{O}(1)$ other parties, though their protocol achieved almost-everywhere agreement~\cite{dwork1986fault} rather than agreement, where $1 - O(\log^{-1}n)$ fraction of the parties reach agreement. 
Subsequent studies have sought to bridge this gap by achieving agreement from almost-everywhere agreement. 
To the best of our knowledge, so far the best communication complexity achievable for agreement in authenticated setting
is $\Tilde{O}(1)$ rounds with $\Tilde{O}(\sqrt{n})$ communication and computation~\cite{gelles2024optimal}. With enhanced cryptographic primitives such as LWE, this can be improved to $\Tilde{O}(1)$ rounds and $poly(\lambda,\log n)$ communication~\cite{boyle2021breaking,fernando2024scalable}.
For asynchronous Byzantine agreements,~\cite{huang2024byzantine} presents the first polynomial-time protocol that achieves optimal resilience $n > 3t$ against a computationally unbounded adversary.

However, for {\em deterministic Byzantine protocols},~\cite{dolev1985bounds} established that achieving agreement inherently requires at least $\Omega(t^2)$ communication. Furthermore,~\cite{civit2024all} extended this lower bound to other generalized validity definitions (e.g., weak validity), demonstrating that any well-defined Byzantine agreement variant necessitates $\Omega(t^2)$ communication. Thus, deterministic Byzantine agreements incur fundamentally more overhead than randomized ones. 

Different from this line of works that focused on single-instance executions of BA protocols, our work investigates compositional executions of BA protocols. Although it remains unknown whether our protocols are optimal in terms of their complexity, they do achieve polynomial communication and computation. Moreover, our refined RMT and RB primitives effectively reduce the communication overhead when transforming single-instance secure BA protocols with long messages. Further improving the complexity of our protocols and proving complexity lower-bounds in our model are interesting open problems.

\smallskip
\textbf{Network Attacks.}
Many prior works on network security have shown that channel attacks (e.g., man-in-the-middle attacks) without corrupting participating parties exist in a wide range of scenarios and pose a significant threat in real applications. By delaying, forging, dropping, or redirecting, attackers can break SSL/TLS certificate validation in many critical software applications \cite{georgiev2012the, sounthiraraj2014smv}, cause security issues for HTTPS and its certificate trust model \cite{clark2013sok}, inflict devastating damage on private and consortium blockchains \cite{ekparinya2018impact}, and so on. 

It is also worth noticing that for transaction messages in DeFi, attackers (miners) can gain Miner Extracted Value (MEV) through transaction reordering attacks such as those studied in \cite{wang2024mvtl}, while \cite{heimbach2023sok} indicated that currently no mitigation schemes can fully solve this problem. 
For a fully connected large-scale network, \cite{tamir2022simple} analyzed Simple Majority Protocol (SMP) under probabilistic message loss and proved that it can reach consensus in three rounds of communication with probability approaching~1.
However, those are within the same consensus protocol and the role of channel attacks in the composition of Byzantine agreements has not been explored yet.

\smallskip
\textbf{Efficient Long Message Dissemination.}
Finally, various asynchronous BA protocols~\cite{abraham2021reaching,gao2021efficient} rely on the broadcast of long messages with sizes typically reaching $O(n \log n)$. 
Reliable dissemination of large blocks is also a fundamental requirement for many consensus protocols~\cite{adam2019aleph,guo2020dumbo,miller2016the}.
Consequently, an important research direction focuses on optimizing the processing of these long messages to effectively reduce communication complexity.

Approaches to this problem primarily rely on coding theory.
For instance, \cite{das2021asynchronous} utilizes Reed-Solomon error-correcting codes to restructure the reliable broadcast process.
In order to satisfy the redundancy constraints in error correction, this method necessitates the additional use of a hash function to verify the correctness of decoding. 
An alternative approach 
shifts from error-correcting codes to erasure-correcting codes, which requires knowing the positions of correct symbols. 
This technique, adopted by \cite{cachin2005asychronous,zhang2025optimistic,shoup2025kudzu}, typically employs Merkle trees to verify the integrity of the erasure-coded fragments. 
\section{Our Model}
\label{sec:model}

\subsection{Byzantine Agreement}
We consider the problem of \textit{Byzantine Agreement} where $P = \{P_1, \dots, P_n\}$ is a set of $n$ parties, and $P$ is common knowledge among all parties. 
The protocol $\Pi$ is executed $m$ times in parallel (or concurrently), where we denote the $k$-th instance by $\Pi_k$.
Each party $P_i$ in each instance $\Pi_k$ is formally modeled by an interactive (randomized) Turing machine $ITM_k^i$, with an input tape containing its initial value, an output tape for its final output, a random tape, and $n-1$ pairs of (input, output) communication tapes corresponding to the other $n-1$ parties in the same $\Pi_k$, denoted by $(Input_{i,j}^k, Output_{i,j}^k)$  for each $P_j \neq P_i$. More specifically, $Output_{i,j}^k$ contains messages that $i$ sends to $j$, and $Input_{i,j}^k$ contains messages that $i$ receives from $j$. 
Under parallel or concurrent composition, a party locally runs $m$ copies of such interactive Turing machines simultaneously, each corresponding to a particular protocol instance.

\smallskip
\smallskip
\textbf{Network and Communication.}
In most parts of the paper before Section~\ref{sec:concurrentBA}, the network is synchronous and point-to-point: communication proceeds in rounds, each consisting of a {\em send} phase followed by a {\em receive} phase. 
In the send phase, a party writes the sending message onto the corresponding output tape. Then, in the receive phase, the message is written onto the target party's input tape under the same protocol instance.

More specifically, the communication channel between parties $P_i$ and $P_j$, denoted by $C_{ij}$, is the set of all input and output tapes between the two parties, for all $m$ protocol instances. By symmetry, $C_{ji}$ means the same thing as $C_{ij}$ and an attack on one is also an attack on the other.
In an uncompromised channel $C_{ij}$, messages originating from $Output_{i,j}^k$ are delivered exclusively to $Input_{j,i}^k$ for all $k\in\{1, ..., m\}$, and vice versa.

\smallskip
\smallskip
\textbf{Stateless Setup.}
In stateless compositions of BA protocols, there is no common session identifier for the protocol instances. 
The indices of different protocol instances are for discussion purposes and are not accessible by the parties. 
Indeed, each Turing machine $ITM^i_k$ has access to its tapes but doesn't know the global index $k$.
Internally, each party may have individual numbering for its Turing machines, but those numbers may not be consistent across different parties.

For authenticated Byzantine Agreement (ABA), a {\em common setup} is shared across all protocol instances. In particular, all parties undergo a one-time trusted preprocessing phase that generates cryptographic primitives,
which are then used consistently across all protocol instances. 
Formally, each interactive Turing machine has a read-only {\em setup tape,} and all Turing machines run by the same party share this setup tape. 
For example, this tape may contain the secret key of a signature scheme used by the party, as well as the public verification keys of other parties.
For unauthenticated Byzantine Agreement, no cryptographic primitives are used and there is no setup tape or a setup phase.

\smallskip
\smallskip
\textbf{Composition of BA Protocols.}
We consider two types of protocol composition: parallel composition and concurrent composition. 
In parallel composition, all protocols start simultaneously, and each round has one time step across all protocol instances. 
In concurrent composition, the adversary can independently control the start time and the time steps of each round of each protocol instance. 

\smallskip
\smallskip
\textbf{Adversarial Model.}
We consider an adversary $\mathcal{A}(t, c)$ operating in the {\em point-to-point full information} setting
in a parallel or concurrent execution of protocol instances $\Pi_1, \Pi_2, ..., \Pi_m$.
Its capabilities include:
\begin{itemize}
    \item {\em Corruption:} The adversary can corrupt up to $t$ parties and chooses which parties to corrupt before any protocol instance begins. This is called static corruption.
    \item {\em Channel Reorder Attack:}  
    The adversary can attack up to $c$ channels between honest parties. 
    When $C_{ij}$ is under attack, the adversary is allowed to redirect messages on corresponding input tapes for different protocol instances, in both directions.
    See below for a formal definition, which is also illustrated in Figure \ref{fig:turing-machine}.
    \end{itemize}

    \begin{definition}[Channel Reorder Attack]
    \label{def:reorder-attack}
    For any two honest parties $P_i$ and $P_j$,
    a channel reorder attack on $C_{ij}$ allows the adversary to arbitrarily redirect messages between different protocol instances and in both directions. 
    Specifically, the adversary can deliver any message $msg \in Output_{x,y}^{k_1}$ to any $Input_{y,x}^{k_2}$, where $k_1, k_2 \in \{1, ..., m\}$ and $\{x, y\} =\{i, j\}$. 
    \end{definition}

    \begin{figure}
    \centering
    \includegraphics[width=0.95\linewidth]{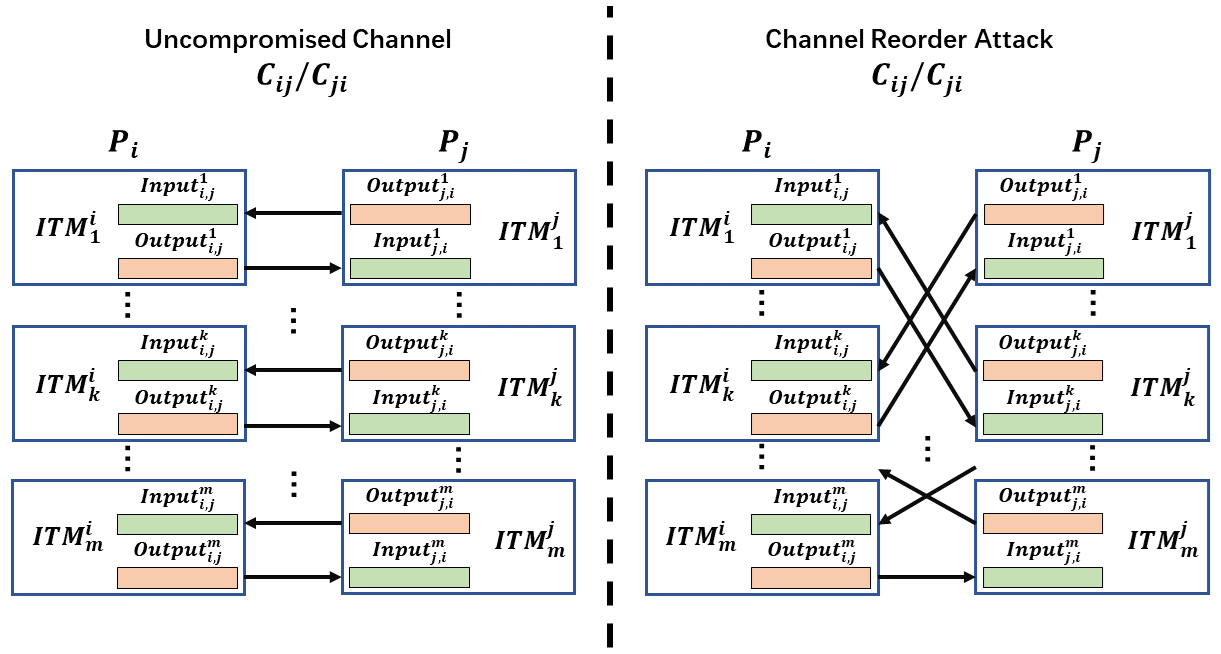}
    \caption{Communication diagrams for channel $C_{ij}/C_{ji}$, with and without channel reorder attack.}
    \label{fig:turing-machine}
    \end{figure}

    \begin{itemize}
    \item {\em Rushing:} The adversary is rushing, meaning it can observe all messages sent by honest parties in a round before deciding the messages to send from corrupted parties and message reordering in the same round.
    
    \item {\em Computation Power:} For our impossibility results, which apply to authenticated BA protocols, we assume that all parties (honest and adversarial) are probabilistic polynomial-time (PPT).
    For our positive results in Sections~\ref{sec:black-box} and~\ref{sec:concurrentBA}, which apply to unauthenticated BA protocols, we allow computationally unbounded adversaries. 
    Finally, for results in Section~\ref{sec:long}, which leverage erasure-correcting codes and cryptographic proof schemes, we again assume polynomial-time adversaries.
\end{itemize}

For this extended adversary model, a protocol $\Pi$ is considered a solution to the BA problem if it satisfies the following properties:

\begin{definition}[Byzantine Agreement]
Let $\mathcal{A}(t, c)$ be an adversary who can corrupt up to $t$ parties and $c$ channels. A protocol $\Pi$ solves the Byzantine Agreement problem if for any adversary $\mathcal{A}(t, c)$, the following two properties hold:
\begin{itemize}
    \item \emph{Agreement}: All honest parties eventually terminate and output the same value.
    \item \emph{Validity}: If all honest parties begin with the same initial value $v$, then their output values must all be $v$.
\end{itemize}
\end{definition}

\subsection{Composition Security}

Composition security plays a pivotal role in analyzing the robustness of cryptographic protocols executed within larger systems or concurrently with other protocol instances, rather than in isolation.
Our work adopts the definition of composition security as presented in~\cite{lindell2006composition}, with the notable distinction that we explicitly consider attacks on communication channels.

\begin{definition}[Composition Security of ABA]
Let $P_1, ..., P_n$ be parties for an ABA protocol $\Pi$. 
We say that $\Pi$ remains secure under $m$ instances of parallel (or concurrent) executions, if for every PPT adversary $\mathcal{A}(t, c)$, the requirement of BA holds for every individual instance of $\Pi$ within the following procedure:
\begin{enumerate}
    \item \emph{Setup Phase:} A single, global trusted setup phase is performed once. This phase generates setup strings $s_1, \dots, s_n$ for parties $P_1, \dots, P_n$, 
    which are then consistently used across all $m$ protocol instances. 

    \item \emph{Static Corruption:} The adversary $\mathcal{A}(t, c)$ statically corrupts up to $t$ parties, gaining full control over their actions.
    Additionally, the adversary can attack up to $c$ communication channels $C_{ij}$ between honest parties, enabling channel reorder attacks described in Definition~\ref{def:reorder-attack}.
    
    \item \emph{Parallel (Concurrent) Executions:} The following procedure is repeated in parallel (or concurrently) for each protocol instance, until the adversary halts:
    \begin{itemize}
        \item The adversary specifies the initial input values for all parties.
        \item Each party uniformly generates the content of its random tape for the execution in this protocol instance.
        \item All parties are invoked for the execution of $\Pi$.
        \item The adversary determines the messages sent by the corrupted parties, while honest parties strictly adhere to the protocol $\Pi$.
        \item Messages transmitted over attacked channels are subject to the adversary's reordering on corresponding input tapes, whereas messages on unattacked channels between honest parties are delivered reliably.
    \end{itemize}
\end{enumerate}
\end{definition}


For unauthenticated Byzantine Agreement protocols, the definition of composition security is analogous, with the difference that no setup phase or cryptographic primitives are needed when considering unbounded adversaries. 
\section{Impossibility Result}
\label{sec:impossibility}

Cryptographic authentication allows protocols to 
tolerate a larger number of corrupted 
parties.
In synchronous networks, authenticated BA can tolerate up to $n/2$ corruptions under a single execution, significantly exceeding the $n/3$ limit for unauthenticated protocols.
However, this intuition breaks down in composed scenarios. When multiple executions occur, a signature from a party $P$ on message $x$ does not prove that $P$ signed $x$ in the ``current'' specific execution.
An adversary can borrow signed messages from one execution and use them in another, rendering the public-key infrastructure not as useful as before in distinguishing execution contexts.

\begin{theorem}
\label{thm:impossibility}
    There does not exist a protocol for stateless authenticated Byzantine Agreement in a synchronous network with $n$ parties that remains secure under parallel composition (even for just two executions) against an adversary $\mathcal{A}(t, c)$, provided that (1) $n \leq 3t$, or (2) $n \leq 2c + 2t + 1$.
\end{theorem}

\begin{proof}[Proof Sketch]
The core idea is to construct two systems.
In each system, the adversary corrupts parties and reconfigures the network topology so that certain instances are cross-connected between parallel executions while receiving conflicting inputs.
We show that, for certain honest instances, the two systems are indistinguishable.
However, due to the agreement and validity properties guaranteed by the BA protocol, this honest instance must output $1$ in one system and $0$ in the other. 
This leads to a contradiction.
The formal rigorous proof is deferred to Appendix \ref{appendix:proof of impossibility}.
\end{proof}

\section{Black-box Compiler for Parallel Composable BA}
\label{sec:black-box}

This section presents a novel black-box compiler that transforms any BA protocol into one that remains secure under parallel composition. 
The proposed approach addresses the challenges introduced by message reordering attacks through the integration of a new Reliable Message Transmission (RMT) protocol.

\subsection{RMT for Parallel Composition}

To achieve security under parallel composition, our compiler replaces the standard point-to-point communication primitives in a given BA protocol with our newly designed Reliable Message Transmission protocol. 
This RMT primitive guarantees correct and unforgeable message delivery within the same protocol instance, even in the presence of adversarial message reordering. 
The formal specification of the RMT is presented in Protocol~\ref{sim:RMT}.%
\footnote{If the values of $t$ and $c$ are publicly known, then the threshold $\frac{n-1}{2}$ in Protocol~\ref{sim:RMT} can be replaced by $t+c$ and the protocol continues to work. The description above works when the values of $t$ and $c$ are unknown.}

\begin{algorithm}
\caption{Reliable Message Transmission}
\label{sim:RMT}
Given that party $P_i$ wishes to send message $msg$ to $P_j$, the reliable message transmission protocol proceeds as follows:
    \begin{itemize}
        \item Round 0: Party $P_i$ initiates this message transmission by sending the message $(deliver, msg, P_i, P_j)$ to all parties except $P_j$.
        \item Round 1: Party $P_i$ sends the message $(deliver, msg, P_i, P_j)$ to $P_j$; and each other party $P_k \neq P_j$ that received message $(deliver, msg, P_i, P_j)$ from $P_i$ in round 0 forwards the same message $(deliver, msg, P_i, P_j)$ to party $P_j$.
        \item Decision: Upon party $P_j$ receiving strictly more than $\frac{n-1}{2}$ copies of the message $(deliver, msg, P_i, P_j)$ by the end of round 1, it accepts message $msg$ from~$P_i$. 
    \end{itemize} 
\end{algorithm}

When $n > 2t + 2c + 1$, this RMT protocol ensures message delivery through a redundancy-based approach.
The principle is to use a majority vote to secure message transmission against potential attacks, a technique also found in prior works but in different formats; see, e.g., ~\cite{rabin1989verifiable,fitzi2000partial,civit2023easy}.
When executed by all participating parties, the RMT primitive is characterized by two essential properties for its security within the black-box compiler.
\begin{itemize}
\item \emph{Correctness}: If both parties $P_i$ and $P_j$ are honest, and $P_i$ used the RMT protocol to send a message $msg$ to $P_j$ in some round $r$, then $P_j$ will accept the message $msg$ from $P_i$ by the end of round $r+1$.
\item \emph{Unforgeability}: If both parties $P_i$ and $P_j$ are honest, and $P_i$ did not use the RMT protocol to send message $msg$ to $P_j$ in round $r$, then $P_j$ will not accept $msg$ from $P_i$ by the end of round $r+1$.
\end{itemize}

The following lemma establishes the security guarantees of the RMT protocol:

\begin{lemma}
    Protocol~\ref{sim:RMT} satisfies both \emph{correctness} and \emph{unforgeability} properties against any adversary $\mathcal{A}(t,c)$, provided that $n > 2c + 2t + 1$.
\end{lemma}

\begin{proof}
    \emph{Correctness:}  Since $P_i$ is honest, it will send $(deliver, msg, P_i, P_j)$ to all other parties except $P_j$ in round 0. 
    In round 1, each party $P_k \ne P_j$, including $P_i$, forwards/sends the message to $P_j$, unless:
    \begin{itemize}
        \item $P_k$ is corrupted by the adversary, or
        \item The communication channel $C_{ik}$ or $C_{kj}$ is attacked by the adversary.
    \end{itemize}
    Since each corruption or channel attack can prevent at most one forwarding, at most $t + c$ of these forwarded messages can be suppressed. 
    Thus, $P_j$ receives at least $n - 1 - t - c$ forwarded messages. 
    Since $n > 2c + 2t + 1$, we have $\frac{n-1}{2} > t + c$, which is equivalent to $n-1-t-c > \frac{n-1}{2}$.
    Therefore, $P_j$ receives strictly more than $\frac{n-1}{2}$ copies of the message $(deliver, msg, P_i, P_j)$, and accepts $msg$ from $P_i$ as valid by the end of round 1.

    \emph{Unforgeability:} If $P_i$ did not initiate the RMT protocol for message $msg$ to $P_j$, then the only way for $P_j$ to receive $(deliver, msg, P_i, P_j)$ is via adversarial intervention. 
    The adversary can inject at most $t + c$ such messages through corrupted parties or attacked channels.
    Since $t + c < \frac{n - 1}{2}$, $P_j$ receives fewer than the majority threshold of messages, and thus does not accept $msg$ from $P_i$.
\end{proof}

\subsection{Black-box Compiler}

Using the RMT protocol, our black-box compiler applies a simple yet effective transformation: given an existing BA protocol $\Pi_{*}$, it constructs a new protocol $\Pi$ by replacing all point-to-point message delivery operations with the RMT protocol.
The internal logic, state transitions, and computation procedures of $\Pi_{*}$ are completely preserved, ensuring that the transformation is non-intrusive and protocol-agnostic. 

Formally, a party $P_i$'s $ITM^i$ in protocol $\Pi$ uses its corresponding $ITM_{*}^i$ in $\Pi_{*}$ as a sub-routine. The outer one, $ITM^i$, handles message sending and receiving according to the RMT protocol, and passes accepted messages to the inner one, $ITM_{*}^i$. The inner one then computes its state transition and out-going messages according to $\Pi_{*}$, and passes the out-going messages to $ITM^i$ to send out.

Notice that in $\Pi$, a message that originally takes one round to deliver in $\Pi_{*}$ now takes two rounds. As such, honest parties will only initiate a message sending in odd rounds (wlog assuming the protocol starts with round 1) 
and will only accept a message in even rounds.
More specifically, if in the execution of $\Pi_{*}$ an honest party $P_i$ sends a message to $P_j$ in round~$r$, then in the execution of $\Pi$, $P_i$ initiates the corresponding message sending in round $2r-1$ and $P_j$ accepts the message by the end of round $2r$.

\begin{theorem}
    \label{thm:parallel}
    Let $\Pi_{*}$ be a BA protocol for $n$ parties that tolerates up to $t$ Byzantine corruptions. 
    Let $\Pi$ be the protocol obtained by replacing all point-to-point communication in $\Pi_{*}$ with the RMT protocol.
    Then, $\Pi$ is secure under parallel composition of any number of executions against any adversary $\mathcal{A}(t, c)$, provided that $n > 2t + 2c + 1$.
\end{theorem}

\begin{proof}
    Arbitrarily fixing an instance $\Pi_j$, we prove its security under parallel composition by induction. Consider a mental game with a single execution of the protocol $\Pi_{*}$ as follows.

    For each honest party $P_i$ and the corresponding $ITM^i_j$ in $\Pi_j$, let its inner Turing machine be $ITM^i_{*j}$. We  construct an $ITM^i_{*}$ for $\Pi_{*}$, such that $ITM^i_{*}$ has the same initial input tape, random tape (and setup tape) as $ITM^i_j$, and thus also as $ITM^i_{*j}$.
    
    Note that, for any adversary $\mathcal{A}(t, c)$ in $\Pi$, the only effects it can have on some party $P_i$ in instance $\Pi_j$ are the following:
    \begin{itemize}
        \item It can fully control $ITM^i_j$ if the party $P_i$ is corrupted.
        \item It can cause an honest party $P_i$'s $ITM^i_j$ to accept a message $msg$ from a corrupted party $P_j$ at some round $2r$. This is because, in $\Pi$, no honest ITM accepts a message in odd rounds, and due to the unforgeability property of RMT, the adversary cannot forge a message from other honest parties to make an honest $P_i$ accept it.
    \end{itemize}
    Therefore, for any adversary $\mathcal{A}$ in $\Pi_j$, we simulate a corresponding adversary $\mathcal{A}_{*}$ in $\Pi_{*}$ as follows:
    \begin{itemize}
        \item Whenever $\mathcal{A}$ corrupts a party $P_i$ in $\Pi_j$, $\mathcal{A}_{*}$ corrupts the corresponding $ITM_{*}^i$ in $\Pi_{*}$.
        \item Whenever $\mathcal{A}$ causes an honest $ITM^i_j$ to accept a message $msg$ from a corrupted party $P_j$ at round $2r$ in $\Pi_j$, $\mathcal{A}_{*}$ sends message $msg$ from $ITM^j_{*}$ to $ITM^i_{*}$ at round $r$ in $\Pi_{*}$.
    \end{itemize}

    Since for every honest party $P_i$, both $ITM^i_{*j}$ and $ITM^i_{*}$ have the same input, random, and setup tapes at the beginning of round 1, they have the same state transition and out-going messages.
    Because $ITM^i_j$ handles message sending and receiving according to the RMT protocol, and by the construction of the adversary $\mathcal{A}_{*}$ in $\Pi_{*}$,
    at the end of round 2 in $\Pi_j$, $ITM^i_j$ has accepted exactly the same messages as received by $ITM^i_{*}$ at the end of round 1 in~$\Pi_{*}$.
    Thus $ITM^i_{*j}$ after round 2 of $\Pi_j$ has the same view as $ITM^i_{*}$ after round 1 of $\Pi_{*}$. That is, $\text{view}_{\Pi_j}^{2}(ITM^i_{*j}) = \text{view}_{\Pi_{*}}^{1}(ITM^i_{*})$. 
    
    Let $\text{round}(\Pi_{*})$ be the maximum finite number of rounds in protocol $\Pi_{*}$.
    For each $r = 1, \dots, \text{round}(\Pi_{*})$, we compare the view of $ITM^i_{*j}$ after $2r$ rounds of $\Pi_j$ with the view of $ITM_{*}^i$ after $r$ rounds of $\Pi_{*}$. 
    By the inductive hypothesis, we have $$\text{view}_{\Pi_j}^{2r-2}(ITM^i_{*j}) = \text{view}_{\Pi_{*}}^{r-1}(ITM^i_{*}).$$
    Thus they again have the same state transition and out-going messages in the corresponding round $2r-1$ and round $r$.
    
    Due to the correctness and the unforgeability properties of the RMT protocol, and by the construction of the adversary $\mathcal{A}_{*}$, again by the end of round $2r$ of $\Pi_j$, $ITM^i_j$ accepts exactly the same messages as received by $ITM^i_{*}$ at the end of round $r$ in $\Pi_{*}$.
    Hence, $\text{view}_{\Pi_j}^{2r}(ITM^i_{*j}) = \text{view}_{\Pi_{*}}^r(ITM^i_{*})$ for all $r$.
    
    Accordingly, we conclude that the final output value of $ITM^i_j$ in $\Pi_j$ (which is that of $ITM^i_{*j}$) and the final output value of $ITM^i_{*}$ in $\Pi_{*}$ are identical. 
    As this holds for all honest parties in $\Pi_j$, we have that $\Pi_j$ inherits the same agreement and validity properties as $\Pi_{*}$. Thus $\Pi$ is secure under parallel composition.
\end{proof}

By applying our compiler to any existing BA protocol, we can construct a protocol that remains secure under parallel composition.
In particular, by transforming the protocol of Garay and Moses~\cite{garay1998fully}, which is secure in the unauthenticated, stateless setting whenever 
$n > 3t$, we obtain the following corollary:

\begin{corollary}
    There exists a protocol for stateless, unauthenticated BA that is secure under parallel composition of any number of executions against any adversary $\mathcal{A}(t, c)$, provided that the total number of parties $n$ satisfies $n > 3t$ and $n > 2c + 2t + 1$.
\end{corollary}

While the black-box compiler significantly enhances BA protocol's security, it introduces some performance overheads.
The round number of the compiled protocol $\Pi$ is doubled compared with the original protocol $\Pi_{*}$,
and the communication overhead per message of $\Pi_{*}$ 
is $O(n)$. In particular, for a message from $P_i$ to $P_j$, $P_i$ sends to $n-2$ parties in Round 0 of RMT, and up to $n-1$ parties may send the message to~$P_j$ in Round~1. 
These overheads are the trade-off for achieving security against adversarial message reordering in parallel compositions.
In Section \ref{sec:long} we discuss how to reduce the communication overhead for long messages.
\section{Black-box Compiler for Concurrent Composition under Asynchronous Network}
\label{sec:concurrentBA}

The preceding sections focused on parallel composition under a synchronous network model, where message delivery is guaranteed within a fixed time bound (i.e., a round).
In this section, we shift our attention to a more challenging setting, the asynchronous network model.
Here, messages can be delayed arbitrarily, and there is no global clock to coordinate the actions of different parties.
Moreover, we consider concurrent composition, where protocol instances may not start simultaneously or proceed at the same pace, making it significantly harder for a party to determine whether a received message indeed belongs to a particular instance or has been reordered.

To address these challenges, we propose a black-box compiler that transforms any standard asynchronous BA protocol into one that is secure under concurrent composition. Our main technical contribution lies in securely constructing the fundamental communication primitives ---Reliable Broadcast (RB) and Reliable Message Transmission (RMT) --- within this demanding environment.

\subsection{Reliable Broadcast Under Concurrent Composition}
RB~\cite{chang1984reliable} is a fundamental communication primitive essential for constructing many BA protocols~\cite{bracha1987asynchronous,berman1992bit,gao2021efficient} in asynchronous networks.
It is designed to ensure that a message sent by an honest sender is reliably delivered to all parties. 
Specifically, RB protocols guarantee the following properties:
\begin{itemize}
    \item \emph{Agreement}: If any honest party accepts a message $m$, then all honest parties eventually accept $m$.
    \item \emph{Validity}: If the sender is honest and broadcasts a message $m$, then all honest parties eventually accept $m$.
\end{itemize}

Based on the work of~\cite{bracha1987asynchronous} and utilizing the idea behind our RMT protocol in Section~\ref{sec:black-box}, we present Protocol~\ref{protocol:broadcast} that achieves reliable broadcast under concurrent executions against any adversary $\mathcal{A}(t,c)$, under the conditions $n > 3t$ and $n > 2c + 2t + 1$.
To the best of our knowledge, this is the first RB protocol to achieve concurrent security within the stateless model, a contribution that may be of independent interest.

The core idea is to avoid accepting messages based solely on direct receipt.
Instead, each party maintains a local view of all other parties' states.
This view is only updated after a majority (i.e., more than $\frac{n-1}{2}$) of confirming messages have been received.
Specifically, let the original sender in the RB protocol be $P_g = P_1$.
The local state of a party $P_j$ with $j=1,\dots,n$ at the view of some party $P_i$, denoted by $S^i_j$, can be in one of the following three forms:
\begin{itemize}
    \item $\{initial, \perp\}$: no value has been received by $P_j$ yet;
    \item $\{prepare, v\}$: value $v$ tentatively accepted by $P_j$;
    \item $\{commit, v\}$: value $v$ confirmed by $P_j$ and will not change anymore.
\end{itemize}
Respectively, the messages sent by a party $P_i$ during the protocol reflect its own state and take one of the following forms: 
\begin{itemize}
\item $(receive, v, P_g)$: $P_i$ has received value $v$ sent by $P_g$;
\item $(echo, v, P_i)$: $P_i$ has tentatively accepted $v$;
\item $(ready, v, P_i)$: $P_i$ has confirmed value $v$, that is, $P_i$ will not accept any message other than $v$.
\end{itemize}

\begin{algorithm}[t]
\caption{Reliable Broadcast}
\label{protocol:broadcast}
Say party $P_g$ wishes to broadcast value $v$.
Each party $P_i$ initializes its local view of each party $P_j$'s state as $S^i_j = \{initial, \perp\}$.
The protocol proceeds as follows:
    \begin{itemize}
        \item Step 0 (performed by $P_g$): send $(initial, v)$ to all parties including itself.
    \end{itemize}
    For each party $P_i$:
    \begin{itemize}
        \item Step 1: Upon receiving some message $(initial, v)$ from $P_g$, send $(receive, v, P_g)$ to all parties. 
        \item Step 2: Wait until the receipt of strictly more than $\frac{n-1}{2}$ $(receive, v, P_g)$ messages for the same value $v$ from different parties. 
        If local state $S^i_i = \{initial, \perp~\}$, set $S^i_i = \{prepare, v\}$ and send $(echo, v, P_i)$ to all parties.
        \item Step 3: Upon receiving some message $(echo, v, P_j)$ from party $P_j$, forward it to all parties.
        \item Step 4: Wait until the receipt of strictly more than $\frac{n-1}{2}$ $(echo, v, P_j)$ messages for the same $v$ and $P_j$ from different parties. 
        If $S^i_j = \{initial, \perp\}$, set $S^i_j = \{prepare, v\}$. 
        \item Step 5: Wait until strictly more than $\frac{2n-1}{3}$ parties $P_j$ have local state $S^i_j =\{prepare, v\}$ for the same $v$.
        If $S^i_i = \{initial, \perp\}$ or $S^i_i = \{prepare, *\}$, where~$*$ can be any value, set $S^i_i = \{commit, v\}$ and send $(ready, v, P_i)$ to all parties.
        \item Step 6: Upon receiving some message $(ready, v, P_j)$ from party $P_j$, forward it to all parties.
        \item Step 7: Wait until the receipt of strictly more than $\frac{n-1}{2}$ $(ready, v, P_j)$ messages for the same $v$ and $P_j$ from different parties. 
        If $S^i_j = \{initial, \perp\}$ or $S^i_j = \{prepare, *\}$, where $*$ can be any value, set $S^i_j = \{commit, v\}$. 
        \item Step 8: Wait until strictly more than $\frac{n-1}{3}$ parties $P_j$ have state $S^i_j = \{commit, v\}$ for the same~$v$.
        If $S^i_i = \{initial, \perp\}$ or $S^i_i = \{prepare, *\}$, where $*$ can be any value, set $S^i_i = \{commit, v\}$ and send $(ready, v, P_i)$ to all parties.
        \item Step 9: Wait until strictly more than $\frac{2(n-1)}{3}$ parties $P_j$ have state $S^i_j = \{commit, v\}$ for the same~$v$. Accept~$v$.
    \end{itemize} 
\end{algorithm}

We now prove that Protocol~\ref{protocol:broadcast} satisfies the properties of reliable broadcast under concurrent composition.
In particular, we consider $m$ instances of Protocol~\ref{protocol:broadcast}
concurrently executed 
against adversary $\mathcal{A}(t, c)$, under conditions $n > 3t$ and $n > 2c + 2t + 1$. 

\begin{lemma}
\label{lemma:message-to-state}
In any instance, for any two honest parties $P_i$ and $P_j$, $P_i$ eventually sets its local state $S^i_j = \{commit, v\}$ if and only if $P_j$ broadcasts the message $(ready, v, P_j)$ during the execution of this instance.
\end{lemma}
\begin{proof}
    If an honest party $P_j$ broadcasts $(ready, v, P_j)$, then in the execution of any other honest party $P_i$, this message will be forwarded by every party $P_k$ unless:
    \begin{itemize}
        \item $P_k$ is corrupted by the adversary, or
        \item the communication channel $C_{jk}$ or $C_{ki}$ is under attack.
    \end{itemize}
    Since each adversary action can block at most one message, the number of suppressed forwards is at most $t + c$. 
    Therefore, $P_i$ eventually receives more than $\frac{n-1}{2}$ such messages and sets its local state for $P_j$ to $S^i_j = \{commit, v\}$.

    Conversely, if $P_j$ did not send the message in that instance, then the adversary can inject at most $t + c$ copies $(ready, v, P_j)$.
    Since $t + c < \frac{n-1}{2}$, an honest party will not receive enough messages to satisfy the threshold, and thus will not set $P_j$'s local state to $S^i_j =\{commit, v\}$.
\end{proof}

\begin{lemma}
\label{lemma:prepare-to-echo}
In any instance, for any two honest parties $P_i$ and $P_j$, $P_i$ eventually sets its local state $S^i_j = \{prepare, v\}$ only if $P_j$ broadcasts the message $(echo, v, P_j)$ during the execution of this instance.
\end{lemma}
\begin{proof}
    If $P_j$ did not send the message $(echo, v, P_j)$  in that instance, then the adversary can inject at most $t + c$ copies $(echo, v, P_j)$.
    Since $t + c < \frac{n-1}{2}$, an honest party will not receive enough messages to satisfy the threshold, and thus will not set $P_j$'s local state to $S^i_j =\{prepare, v\}$.
\end{proof}

\begin{lemma}
\label{lemma:ready-equal}
    In any instance, for any two honest parties $P_i$ and $P_j$ who respectively send messages $(ready, v, P_i)$ and $(ready, u, P_j)$ during the execution of this instance, we have $v = u$.
\end{lemma}
\begin{proof}
    Assume for contradiction.
    Let $P_x$ and $P_y$ be the first honest parties to send $(ready, v, P_x)$ and $(ready, u, P_y)$, respectively. 
    Party $P_x$ must have observed strictly more than $\frac{2n-1}{3}$ party $P_j$ with local state $S^x_j=\{prepare, v\}$.
    Party $P_y$ must have observed strictly more than $\frac{2n-1}{3}$ party $P_j$ with local state $S^y_j = \{prepare, u\}$.
    Because $t < \frac{n}{3}$, this implies an overlap where an honest party $P_k$ has $\{prepare, v\}$ in the view of $P_x$ and $\{prepare, u\}$ in the view of $P_y$.

    Therefore, by Lemma~\ref{lemma:prepare-to-echo}, it implies that $P_k$ sent both $(echo, v)$ and $(echo, u)$ messages.
    However, an honest party sends only one type of \emph{echo} message during a reliable broadcast protocol.
    Therefore, $v = u$.
\end{proof}

\begin{lemma}
\label{lemma:accept-equal}
    In any instance, for any two honest parties $P_i$ and $P_j$ who accept values $v$ and $u$ respectively, we have $v = u$.
\end{lemma}
\begin{proof}
    If $P_i$ accepts $v$, then it must have observed at least $\lfloor \frac{2(n-1)}{3}\rfloor+1$ parties $P_k$ with the state $S^i_k = \{commit, v\}$ in its local view, which includes at least $\lfloor \frac{n-1}{3}\rfloor+1$ honest parties. 
    Similarly, $P_j$ must have observed at least $\lfloor \frac{2(n-1)}{3}\rfloor+1$ parties $P_k$ with state $S^j_k = \{commit, u\}$ in its local view, which includes at least $\lfloor \frac{n-1}{3}\rfloor+1$ honest parties.
    By Lemma~\ref{lemma:message-to-state}, there exists at least one honest party $P_x$ that have sent $(ready, v, P_x)$, and one honest party $P_y$ that have sent $(ready, u, P_y)$.
    Applying Lemma~\ref{lemma:ready-equal}, it follows that $v = u$.
\end{proof}

\begin{lemma}
\label{lemma:honest-accept}
    In any instance, if an honest party $P_i$ accepts a value $v$, then every other honest party eventually accepts $v$ in that instance.
\end{lemma}
\begin{proof}
    Suppose $P_i$ accepts value $v$.
    Then its local view must include at least $\lfloor \frac{2(n-1)}{3}\rfloor+1$ parties $P_x$ with the state $S^i_x = \{commit, v\}$, including at least $\lfloor \frac{n-1}{3}\rfloor+1$ honest parties.
    By Lemma~\ref{lemma:message-to-state}, every honest party $P_k$ in these $\lfloor \frac{n-1}{3}\rfloor+1$ parties must have sent $(ready, v, P_k)$ in that instance.
    Also by Lemma~\ref{lemma:message-to-state}, eventually, all honest parties $P_j$ will receive these messages and update their local views, setting $P_k$'s state to $S^j_k = \{commit, v\}$.
    According to step 8 of the protocol, upon observing $\lfloor \frac{n-1}{3}\rfloor+1$ parties in state $\{ commit, v \}$, every honest party $P_j$ updates its own state to $S^j_j=\{commit, v\}$, and send $(ready, v, P_j)$ to all parties.
    Again, by Lemma~\ref{lemma:message-to-state}, all honest parties $P_y$ will receive these messages and update their local views, setting $P_j$'s state $S^y_j=\{commit, v\}$ eventually.
    Since $n > 3t$, the number of honest parties exceeds $\lfloor \frac{2(n-1)}{3}\rfloor+1$.
    Thus, every honest party will eventually have at least $\lfloor \frac{2(n-1)}{3}\rfloor+1$ $\{commit, v\}$ states in its local view and will accept $v$.
\end{proof}

\begin{lemma}
\label{lemma:broadcast-accept}
    In any instance, if an honest party $P_g$ broadcasts a value $v$, then all honest parties eventually accept $v$ in that instance.
\end{lemma}
\begin{proof}
    Since $P_g$ is honest, it sends $(initial, v, P_g)$ to all parties.
    At least $(n-t-c)$ honest parties receive this message and forward $(receive, v, P_g)$ to all other parties.
    Each honest party $P_i$ eventually receives more than $\frac{n-1}{2}$ such $(receive, v, P_g)$ messages.
    It then sets its local state $S^i_i = \{prepare, v\}$ and sends $(echo, v, P_i)$ to all parties.
    As a result, each honest party $P_i$ receives more than $\frac{n-1}{2}$ \emph{echo} messages for value $v$ and each honest party $P_j$, and thus sets the state $S^i_j = \{prepare, v\}$.
    This ensures that every honest party $P_i$ observes enough \emph{prepare} state in its local view to satisfy the condition for committing.
    Accordingly, each honest party updates its own state $S^i_i=\{commit, v\}$ and sends $(ready, v, P_i)$ to all parties.
    
    In Step 7 of the protocol, each honest party $P_i$ receives more than $\frac{n-1}{2}$ \emph{ready} messages for value $v$ and each honest party $P_j$, causing each party to update the state of $P_j$ to $S^i_j=\{commit, v\}$.
    Therefore, since $n>3t$, every honest party $P_i$ eventually has at least $\lfloor \frac{2(n-1)}{3}\rfloor+1$ $\{commit, v\}$ states in local view and accepts $v$.
\end{proof}

\begin{theorem}\label{thm:concurrent-RB}
    Protocol~\ref{protocol:broadcast} realizes RB against any adversary $\mathcal{A}(t, c)$ under concurrent composition in an asynchronous network, provided that $n > 3t$ and $n > 2t + 2c + 1$; and for any message $M$, it has communication complexity $O(n^3 |M| + n^3 \log n)$.
\end{theorem}
\begin{proof}
    \emph{Validity:} In any instance, if the sender $P_g$ is honest and broadcast value~$v$, then by Lemma~\ref{lemma:broadcast-accept}, all honest parties eventually accept $v$ in that instance.
    
    \emph{Agreement:} In any instance, if the sender $P_g$ is corrupted, and some honest party~$P_i$ accepts a value $v$ in that instance, then by Lemma~\ref{lemma:honest-accept}, all other honest parties eventually accept the same value~$v$ in that instance.

    \emph{Complexity:} 
    The communication complexity is derived from the message lengths and the total number of transmissions.
    The length of any single transmitted message is determined by the message length and party identifiers, which requires $O(|M| + \log n)$ bits.
    The protocol involves $O(n^3)$ messages in total: in Steps 3 and~6, each of the $n$ parties may forward messages from each of the other $n$ parties to all $n$ parties, contributing $O(n^3)$ messages.
    Therefore, the overall complexity is the product of the message count and length: $O(n^3) \cdot O(|M| + \log n) = O(n^3|M| + n^3 \log n)$.
\end{proof}

By adapting the Byzantine agreement protocol $\Pi^*$ from Figure~4 of~\cite{bracha1987asynchronous} (see Appendix~\ref{appendix:bracha-BA} of this paper for completeness) and replacing all broadcast procedures with our RB protocol, we can construct an unauthenticated BA protocol that is secure under concurrent composition in an asynchronous network.
In general, any BA protocol whose security relies solely on the underlying broadcast primitive can be transformed into a compositionally secure protocol in a black-box manner.

\subsection{Reliable Message Transmission Under Concurrent Composition}

Having established the RB primitive for transforming specific types of BA protocols, we now turn our attention to a more fundamental level of communication.
Specifically, to address asynchronous BA protocols that rely on point-to-point messaging, we propose a black-box compiler based on the RMT primitive.
This compiler transforms any existing asynchronous BA protocol into a compositionally secure one. Indeed, one can also use the RMT primitive to simulate message broadcast by first replacing broadcast with point-to-point communication.

Building on Section~\ref{sec:black-box}, our RMT primitive emphasizes eventual delivery. 
In an asynchronous network, this ensures that messages from honest nodes are never lost, despite of arbitrary delays.
The security guarantees of the RMT primitive in an asynchronous network are characterized by the following two properties: 

\begin{itemize}
\item \emph{Correctness}: If both parties $P_i$ and $P_j$ are honest, and $P_i$ initiates the protocol to send a message $M$ to $P_j$, then $P_j$ will \emph{eventually} accept the message $M$ from $P_i$.
\item \emph{Unforgeability}: If both parties $P_i$ and $P_j$ are honest, and $P_i$ does not initiates the protocol to send message $M$ to $P_j$, then $P_j$ will never accept $M$ from $P_i$.
\end{itemize}

We present Protocol~\ref{protocol:RMT-concurrent}, which achieves reliable message transmission under concurrent executions. 
Unlike the round-driven Protocol~\ref{sim:RMT}, this construction is entirely event-driven and ensures delivery through a redundancy-based approach.
The desired properties above hold against any adversary $\mathcal{A}(t,c)$ provided that $n > 2c + 2t + 1$.

\begin{algorithm}[t]
\caption{Reliable Message Transmission - Concurrent Composition}
\label{protocol:RMT-concurrent}
Say party $P_i$ wishes to send a message $M$ to party $P_j$.
The protocol proceeds as follows:
    \begin{itemize}
        \item Step 0 (performed by sender $P_i$): send $(deliver, M, P_i, P_j, P_i)$ to all parties. 
    \end{itemize}
    For each party $P_k$ ($P_k \ne P_j$): 
    \begin{itemize}
        \item Step 1 (performed by each $P_k$): Upon receiving message $(deliver, M, P_i, P_j, P_i)$ from $P_i$, send $(deliver, M, P_i, P_j, P_k)$ to $P_j$. 
    \end{itemize} 
    For receiver $P_j$:
    \begin{itemize}
        \item Step 2: Maintain a message box $\mathcal{T}$, upon receiving $(deliver, M, P_i, P_j, P_k)$ from~$P_k$, add it to $\mathcal{T}$.
        \item Step 3: Upon seeing a subset of $\mathcal{T}$, denoted by $\mathcal{M} =$ \\ $\{(deliver, M, P_i, P_j, P_{k_1}), \ldots, (deliver, M, P_i, P_j, P_{k_m})\}$, such that:
            \begin{enumerate}
                \item The size of $\mathcal{M}$ is strictly greater than $\frac{n-1}{2}$, i.e., $m > \frac{n-1}{2}$.
                \item The set of senders $S = \{P_{k_1}, \ldots, P_{k_m}\}$ are all distinct parties, i.e., the last element $P_{k_j}$ in each message is unique within $\mathcal{M}$.
                \item All messages in $\mathcal{M}$ share the same first four components $(deliver, M, P_i, P_j)$.
            \end{enumerate}
            $P_j$ accepts message $M$ from $P_i$.
    \end{itemize}   
\end{algorithm}

The following lemma establishes the security guarantees of the RMT protocol:

\begin{lemma}
    Protocol~\ref{protocol:RMT-concurrent} 
    realizes RMT
    against any adversary $\mathcal{A}(t,c)$ under concurrent composition, provided that $n > 2c + 2t + 1$; and for any message $M$, it has communication complexity $O(n|M| + n \log n)$.
\end{lemma}

\begin{proof}
    \emph{Correctness:}  Since $P_i$ is an honest party, it performs Step 0 and sends the message $(deliver, M, P_i, P_j, P_i)$ to all $n-1$ parties. Given the asynchronous network model, we rely on the eventual delivery of messages from honest parties over unattacked channels.
    Consider the $n-1$ parties $P_k \ne P_i$. Each honest party $P_k$ that receives the initial message from $P_i$ via an unattacked $C_{ik}$ will perform Step 1 and send $(deliver, M, P_i, P_j, P_k)$ to $P_j$. This message will eventually arrive at $P_j$ if $C_{kj}$ is unattacked.
    A forwarding message to $P_j$ can be suppressed or indefinitely delayed in only two ways:
    \begin{itemize}
        \item The forwarding party $P_k$ is corrupted by the adversary ($t$ possible corruption), or
        \item The channel $C_{ik}$ (Step 0) or $C_{kj}$ (Step 1) is attacked by the adversary. With~$c$ channel attacks available, this can suppress at most $c$ distinct forwarded messages to $P_j$.
    \end{itemize}
    The total number of forwarded messages that can be suppressed is at most $t + c$.

    Since there are $n-1$ potential forwarders, the number of distinct forwarded messages $(deliver, M, P_i, P_j, P_k)$ that eventually reach the honest receiver $P_j$ is at least $n - 1 - (t + c)$.
    Given the condition $n > 2c + 2t + 1$, we can rearrange it to show:
    $$n - 1 - (t + c) > (t + c) \implies n - 1 - (t + c) > \frac{n-1}{2}$$
    Therefore, $P_j$ eventually receives strictly more than $\frac{n-1}{2}$ distinct forwarded messages, and by Step 3, $P_j$ eventually accepts $M$ from $P_i$.

    \emph{Unforgeability:} If $P_i$ did not initiate the RMT protocol for message $M$ to $P_j$, the only way for $P_j$ to receive any message of the form $(deliver, M, P_i, P_j, \_)$ is via adversarial intervention.
    The adversary $\mathcal{A}(t,c)$ can create and send this message through two mechanisms:
    \begin{itemize}
        \item Corrupting up to $t$ parties $P_k$ and forcing them send $(deliver, M, P_i, P_j, P_k)$ to $P_j$. 
        \item Directly injecting up to $c$ messages $(deliver, M, P_i, P_j, P_k)$ over $c$ attacked channels leading to $P_j$.
    \end{itemize}
    The maximum number of distinct messages that the adversary can forge and ensure their delivery to the honest $P_j$ is at most $t + c$.
    Since $n > 2c + 2t + 1$, we have $t + c < \frac{n-1}{2}$.
    The number of forged messages $t + c$ is strictly less than the required acceptance threshold of $\frac{n-1}{2}$.
    Thus, $P_j$ will never accumulate the required majority of forwarded messages to satisfy the condition in Step 3, and therefore never accepts~$M$ from $P_i$.

    \emph{Complexity:} In the protocol, the length of any single transmitted message, which has the format $(deliver, M, P_i, P_j, P_k)$, is determined by $M$ and the three party identifiers. 
    Since the message length is $|M|$ and each party requires $O(\log n)$ bits for unique identification, the length of one message is $O(|M| + \log n)$. 
    The protocol involves $O(n)$ messages in total: $n-1$ messages for the initial broadcast (Step 0) and at most $n-2$ messages for forwarding (Step 1). Therefore, the overall complexity is the product of the message count and length: $O(n) \cdot O(|M| + \log n) = O(n|M| + n \log n)$.
\end{proof}

With the security properties of the RMT primitive established, we now proceed to prove the security of the compiled protocol.

\begin{theorem}
    \label{thm:concurrent-RMT}
    Let $\Pi_{*}$ be an asynchronous BA protocol for $n$ parties that tolerates up to $t$ Byzantine corruptions. 
    Let $\Pi$ be the protocol obtained by replacing all point-to-point communication in $\Pi_{*}$ with the RMT protocol.
    Then, $\Pi$ is secure under concurrent composition of any number of executions against any adversary $\mathcal{A}(t, c)$, provided that $n > 2t + 2c + 1$.
\end{theorem}

\begin{proof}
    We fix an arbitrary execution instance $\Pi_j$ of the compiled protocol. 
    We prove its security by simulation, comparing the execution of $\Pi_j$ (the real world) with an execution of the original protocol $\Pi_{*}$ (the ideal world).
    
    We model the execution of each honest party $P_i$ in $\Pi_j$ as an Interactive Turing Machine, denoted as $ITM^i_j$.
    Since the compiled protocol $\Pi$ is obtained by replacing all point-to-point communication in the original protocol $\Pi_{*}$ with the RMT protocol, the state of $ITM^i_j$ effectively maintains a simulated execution of $\Pi_{*}$.
    We denote this virtualized execution as the \emph{inner Turing machine}, denoted $ITM^i_{*j}$.
    In an ideal world, we construct an $ITM^i_{*}$ for $\Pi_{*}$, such that $ITM^i_{*}$ has the same initial input tape, random tape (and setup tape) as $ITM^i_j$, and thus also as $ITM^i_{*j}$.

    Now we prove that the state transition of the inner machine $ITM^i_{*j}$ in the real execution should be identical to that of a party $ITM^i_{*}$ executing the original protocol $\Pi_{*}$ natively.
    To prove this, we construct a simulated adversary $\mathcal{A}_{*}$ for $\Pi_{*}$ that simulates the behavior of the real-world adversary $\mathcal{A}$ for $\Pi_j$.
    The adversary $\mathcal{A}$ interacts with $\Pi_j$ by corrupting parties, injecting messages, and controlling message scheduling. $\mathcal{A}_{*}$ translates these actions to $\Pi_{*}$ as follows:
    \begin{itemize}
        \item \emph{Corruption:} If $\mathcal{A}$ corrupts party $P_i$ in $\Pi_j$, $\mathcal{A}_{*}$ corrupts the corresponding party in $\Pi_{*}$.
        \item \emph{Message Injecting:} If $\mathcal{A}$ schedules the network such that an honest party $P_i$ in $\Pi_j$ \emph{accepts} a message $M$ from corrupted party $P_k$ (via the RMT acceptance condition) at time $T$, $\mathcal{A}_{*}$ delivers the direct message $M$ from $P_k$ to $P_i$ in $\Pi_{*}$, and let $P_i$ receive the message at time $T$.
        \item \emph{Message Scheduling:} If $\mathcal{A}$ prevents the RMT protocol from completing (i.e., $P_j$ has not yet received enough deliver to accept $M$ from $P_i$), $\mathcal{A}_{*}$ simply delays the delivery of message $M$ in $\Pi_{*}$.
    \end{itemize}

    Since the network is asynchronous, we cannot induce on rounds like in the proof of Theorem~\ref{thm:parallel}. 
    Instead, we induce on the sequence of message acceptance events at honest parties.
    Let both systems start with identical inputs, random tapes, and setup tapes.
    Consequently, the initial state of the inner machine $ITM^i_{*j}$ in $\Pi_j$ is identical to $ITM^i_{*}$ in $\Pi_{*}$.
    Since for every honest party $P_i$, both $ITM^i_{*j}$ and $ITM^i_{*}$ have the same input, random, and setup tapes at the beginning of the protocol, they have the same start state and the first out-going messages.
    That is, $\text{view}_{\Pi_j}^{0}(ITM^i_{*j}) = \text{view}_{\Pi_{*}}^{0}(ITM^i_{*})$.

    Assume that for the first $s$ message acceptance events across the entire network, the time of messages accepted and the resulting internal states of all honest parties are identical in both $\Pi_j$ and $\Pi_{*}$.
    That is, $\text{view}_{\Pi_j}^{s}(ITM^i_{*j}) = \text{view}_{\Pi_{*}}^{s}(ITM^i_{*})$.
    Consider the $(s+1)$-th event where an honest party $P_i$ accepts a message $M$ from sender $P_k$ in $\Pi_j$:
    \begin{itemize}
        \item \textbf{Case 1: Sender $P_k$ is honest.}
        By the inductive hypothesis, $P_k$'s state is identical in both worlds when it generated $M$. Thus, $P_k$ initiates the RMT for $M$ in $\Pi_j$ if and only if $P_k$ sends $M$ in $\Pi_{*}$.
        By the \emph{Correctness} of RMT, if $M$ delivery in $\Pi_j$, $P_i$ eventually accepts $M$. 
        Assume $P_i$ accepts $M$ in time $T^{s+1}$, the simulated adversary $\mathcal{A}_{*}$ will delivers $M$ in $\Pi_{*}$ at $T^{s+1}$ simultaneously.
        Meanwhile, by \emph{Unforgeability} and the simulation process, $\mathcal{A}_{*}$ will not let $P_i$ accept any other message before time $T^{s+1}$.
        Thus, $P_i$ accepts message $M$ from $P_k$ is also the $(s+1)$-th event in $\Pi_*$.

        \item \textbf{Case 2: Sender $P_k$ is corrupted.}
        We assume $P_i$ accepts $M$ in time $T^{s+1}$.
        Due to \emph{Unforgeability}, $P_i$ in $\Pi_j$ only accepts $M$ if the adversary explicitly constructs the necessary echoes.
        Therefore, in $\Pi_{*}$, by the simulation process, $\mathcal{A}_{*}$ will let the corrupted sender send $M$, and $P_i$ will accept $M$ at the same time $T^{s+1}$.
        Thus, the inputs to $P_i$'s state transition function are identical.
    \end{itemize}

    In both cases, upon processing the $(s+1)$-th message, the internal state of $ITM^i_{*j}$ updates exactly as $ITM^i_{*}$ updates. Therefore, 
    $\text{view}_{\Pi_j}^{s+1}(ITM^i_{*j}) = \text{view}_{\Pi_{*}}^{s+1}(ITM^i_{*})$. 

    Since the internal states and views of honest parties in $\Pi_j$ are indistinguishable from those in $\Pi_{*}$, the output values of the honest parties must be identical.
    Therefore, $\Pi_j$ inherits the Agreement and Validity properties of $\Pi_{*}$. 
    As this holds for any arbitrary instance $j$, $\Pi$ is secure under concurrent composition.
\end{proof}

With the RMT compiler of Protocol~\ref{protocol:RMT-concurrent}, we can directly compile Bracha's protocol~\cite{bracha1987asynchronous} to obtain a concurrently secure protocol, by replacing all broadcasts in the original protocol with point-to-point communication and then with our RMT protocol.
For any message length $|M|$, Bracha's protocol in a single phase (which is an RB protocol) has communication complexity $O(n^2 |M|)$, while the compiled protocol in a single phase (which is a concurrently secure RB protocol) has communication complexity $O(n^3 |M| + n^3 \log n)$.
\section{Scalable RB and RMT for Long Messages}
\label{sec:long}

In the preceding section, we presented concurrently composable RB and RMT primitives with communication complexities of $O(n^3 |M| + n^3 \log n)$ and $O(n|M| + n \log n)$, respectively.
When the message size $|M|$ is large, these primitives introduce significant overheads.
To achieve more efficient RB and RMT protocols for long messages while preserving concurrent composability, we leverage techniques from erasure-correcting codes (ECC) and Sequential Proof Schemes.
Our refined protocols achieve the following communication complexity guarantees (summarized in Table \ref{tab:complexity_results}): 
\begin{itemize}
    \item For long messages where $|M| \in \Omega(\lambda \log n)$, RMT achieves $O(|M| + \lambda n \log n)$ and RB achieves $O(n^2|M| + \lambda n^3 \log n)$; and 
    \item For very long messages where $|M| \in \Omega(\lambda n \log n)$, RB further improves to $O(n|M| + \lambda n^3 \log n)$.
\end{itemize}
The security of these refined protocols relies on the collision resistance of cryptographic hash functions, which naturally assumes polynomial-time adversaries instead of unbounded ones.


\subsection{Erasure-Correcting Codes and Proof Schemes}

We define ECC and its properties below, adopting a notation consistent with recent works~\cite{das2021asynchronous,zhang2025optimistic}.

\begin{definition}[Erasure-Correcting Code]
\label{def:ECC}
    Given a Galois field $\mathbb{F}=\mathrm{GF}(2^l)$ with $|\mathbb{F}| > n$, we define each symbol in the code as an element in $\mathbb{F}$. An $(n,d)$-erasure-correcting code consists of two polynomial-time algorithms, $\mathsf{ECCenc}$ and $\mathsf{ECCdec}$, defined as follows:
    \begin{itemize}
        \item $\mathsf{ECCenc}(T) \rightarrow T' = [s'_1, \ldots, s'_n] \in \mathbb{F}^n$: 
        The encoding algorithm takes as input a sequence of $d$ source symbols $T = [s_1, \ldots, s_d] \in \mathbb{F}^d$, and outputs a sequence $T'$ of $n$ symbols.
        \item $\mathsf{ECCdec}(T') \rightarrow T = [s_1, \ldots, s_d] \in \mathbb{F}^d$: 
        The decoding algorithm takes as input a sequence of symbols $\tilde{T'} \in (\mathbb{F} \cup \{\perp\})^n$, where $\perp\not\in \mathbb{F}$ represents an erased symbol, and outputs the sequence of source symbols $T$.
    \end{itemize}
    The scheme must satisfy the correctness property:
    \begin{itemize}
        \item Correctness: For any $T \in \mathbb{F}^d$, let $T' = [s'_1, \ldots,s'_n] = \mathsf{ECCenc}(T)$. For any sequence of symbols $\tilde{T}' = [\tilde{s}'_1, \ldots, \tilde{s}'_n]$, where $\tilde{T}'$ has at least $d$ non-erased symbols $\tilde{s}'_i \neq \perp$ and these symbols all satisfy $\tilde{s}'_i = s'_i$, it has      $\mathsf{ECCdec}(\tilde{T'}) = T$.
    \end{itemize}
The parameter $d$ is also referred to as the minimum decoding threshold.
\end{definition}

As noted in~\cite{li2005the}, standard Reed-Solomon codes~\cite{reed1960Polynomial} provide an efficient instantiation of such ECC, assuming the locations of erased symbols are known. By utilizing ECC, we can partition a large message into $n$ smaller shares and distribute them among parties, significantly lowering the communication overhead. 
We demonstrate the application of these algorithms in the following message segmentation scenario, which will be used as building blocks in later constructions.

\smallskip

{\bf Message Segmentation Example.}
    Consider a system with $n$ participants. To disseminate a message $M$, we select a symbol size $l = \log n$ such that the Galois field size is $2^l = n$. For any message $M$ with length $|M|$ that satisfies $\log n \leq|M| \leq \frac{n}{2}\log n$, we segment the message into $d = \lceil \frac{|M|}{\log n} \rceil$ pieces. More precisely, the process goes in two phases:
    \begin{itemize}
        \item Message Encode: 
        The message $M$ is segmented into a sequence of $d$ symbols, denoted by $M^* = [m_1, m_2, \cdots, m_d]$, with $M^*$ being the concatenation of $m_1, m_2, \dots, m_d$. The process uses an $(n, d)$-ECC. It invokes $\mathsf{ECCenc}(M^*)$ to produce a symbol sequence $S = [s_1, \dots, s_n]$ and sends each $s_i$ to player $P_i$.
        \item Message Decode: 
        Upon receiving a sequence of symbols $S' = [s'_1, \dots, s'_n]$ with at least $d$ symbols that are not $\bot$, the process invokes $\mathsf{ECCdec}(S')$ to recover~$M^*$, which is then reassembled to construct $M$.
    \end{itemize}
Note that we need to segment $M$ into an integer number of symbols, and each symbol has a fixed length. 
Therefore we will pad $M$ to the length of $\lceil \frac{|M|}{\log n} \rceil \cdot \log n$ in practice. 

\smallskip
{\bf Proof Schemes.}
Standard erasure decoding requires explicit knowledge of which symbols are valid. 
In a Byzantine setting, however, the adversary may corrupt players and send wrong symbols without declaring them as ``erased''. 
Therefore we need a scheme to verify the integrity of each symbol, essentially converting Byzantine faults into identifiable erasures. 
We define this mechanism as a {\em sequential proof scheme,} similar to the formulation in~\cite{baric1997collision,zhang2025optimistic}. 

\begin{definition}[Sequential Proof Scheme]
    A Sequential Proof Scheme allows for the generation and verification of cryptographic proofs for elements within a sequence while asserting their locations. Let each element belong to some domain $\mathcal{B}$. 
    Such a scheme consists of two polynomial-time algorithms $\mathsf{PGen}$ and $\mathsf{PVer}$, such that:
    \begin{itemize}
        \item $\mathsf{PGen}(B) \rightarrow (\Pi,
        r)$: 
        Given a sequence $B = [b_1, \dots, b_n] \in \mathcal{B}^n$, it outputs a vector of proofs $\Pi = [\pi_1, \dots, \pi_n]$ corresponding to each element, and a global commitment root $r$.
        \item $\mathsf{PVer}(i, b, \pi,r) \rightarrow \{True,False\}$: Given an index $i \in [n]$, an element $b \in \mathcal{B}$, a proof $\pi$ and a root $r$, it validates whether $b = b_i$ or not, where $b_i$ is the $i$-th element in the sequence committed by $r$.
    \end{itemize}
    The scheme must satisfy the following properties:
    \begin{itemize}
        \item Completeness: For any sequence $B\in \mathcal{B}^n$, if $(\Pi, r) = \mathsf{PGen}(B)$, then for all $i \in [n]$, it holds that 
        $\mathsf{PVer}(i, b_i, \pi_i, r) = \text{True}$.
        \item Collision-free: Given a triple $(B, \Pi, r)$ such that $(\Pi, r) = \mathsf{PGen}(B)$, for any polynomial-time adversary, it is computationally infeasible to find an index $i \in [n]$, an element 
        $b' \neq b_i$ and a proof $\pi'$ such that $\mathsf{PVer}(i, b', \pi', r) = \text{True}$.
    \end{itemize}
\end{definition}

It is worth noting that the sequential proof scheme is not a deterministic primitive: the collision-free property requires inherent randomness in the setup of the scheme, exemplified by the random selection of a hash function from a universal hash family. 
Therefore, a random key $rk$ must be drawn from the keyspace $\mathcal{K}(\lambda, n)$ where $\lambda$ is the security parameter and $n$ is the number of elements in the sequence. 
While both the generation algorithm $\mathsf{PGen}$ and the verification algorithm $\mathsf{PVer}$ are parameterized by~$rk$, this dependency is notationally omitted for the sake of brevity.

By integrating a Sequential Proof Scheme with ECC, we will generate a proof $\pi_i$ for each ECC symbol $s_i$, allowing receivers to discard corrupted symbols and treat them as erasures. 
A common instantiation of such a scheme is by using Merkle tree~\cite{merkle1989certified}, which in turn relies on a collision-resistant hash function $H: \mathcal{B} \to \{0, 1\}^\lambda$, with the output size $\lambda$ being the security parameter. The Merkle tree generates a specific proof for each leaf node, which can be verified together with its position $i$ and root $r$.
Under this construction, the proof size for each symbol is $O(\lambda \log n)$, while the size of the global commitment $r$ (i.e., the Merkle root) is $O(\lambda)$. 

\subsection{Reliable Message Transmission for Long Messages}
\label{subsec:long for RMT}

We present our RMT protocol for long messages in 
Protocol~\ref{protocol:RMT-long-concurrent}.
This protocol is applicable to any message $M$ with $|M| \geq \log n$, and starts to be better than the communication complexity of Protocol~\ref{protocol:RMT-concurrent} when $|M| \in \Omega(\lambda \log n)$.

The construction utilizes an $(n-1, d)$-ECC, where the minimum decoding threshold $d$ is defined as follows. 
Firstly, let $d_{max} = \lceil\frac{n-1}{2}\rceil$ and $\sigma =  \lceil \frac{|M|}{d_{max} \log n } \rceil$.
Then, let the symbol size be $\sigma \log n$ and $d = \lceil \frac{|M|}{\sigma \log n} \rceil \leq d_{max}$.
%
In particular, the protocol segments a message $M$ into at most $d_{max}$ symbols with symbol size at least $\log n$.



\begin{algorithm}[t]
\caption{RMT for Long Message - Concurrent Composition}
\label{protocol:RMT-long-concurrent}
Say party $P_i$ wishes to send a message $M$ to party $P_j$. 
Let $d_{max} = \lceil\frac{n-1}{2}\rceil$ and $\sigma = \lceil \frac{|M|}{d_{max} \log n } \rceil$. Let the symbol size be $\sigma \log n$ (thus the field size is $|\mathbb{F}| = 2^{\sigma }n$) and $d = \lceil \frac{|M|}{\sigma \log n} \rceil$.
Slightly overloading the notation, $M$ also represents the sequence of $d$ symbols whose concatenation reproduces $M$.
The protocol proceeds as follows:
    \begin{itemize}
        \item Step 0 (performed by sender $P_i$): Let $M^* := [m_1, m_2, \ldots, m_{n-1}] := \mathsf{ECCenc}(M)$. 
        Compute $(\Pi, r) := ([\pi_1, \ldots, \pi_{n-1}], r) := \mathsf{PGen}(M^*)$.
        Send message $(deliver, m_{k'}, \pi_{k'}, P_k, r, P_i, P_j)$ to each party $P_k\neq P_i$, where $k'=k$ if $k<i$ and $k'=k-1$ if $k>i$.
    \end{itemize}
    For each party $P_k$:
    \begin{itemize}
        \item Step 1: Upon receiving message $(deliver, m_{k'}, \pi_{k'}, P_k, r, P_i, P_j)$ from $P_i$, forward it to $P_j$. 
    \end{itemize} 
    For party $P_j$:
    \begin{itemize}
        \item Step 2: Maintain a message box $\mathcal{T}$.\\
        Upon receiving $(deliver, m_{k'}, \pi_{k'}, P_k, r, P_i, P_j)$ from $P_k$, let $k'=k$ if $k<i$ and $k'=k-1$ if $k>i$. If $\mathsf{PVer}(k', m_{k'}, \pi_{k'}, r) = True$, add it to $\mathcal{T}$.
        \item Step 3: Upon seeing a subset of $\mathcal{T}$, denoted by $\mathcal{M} =$ \\$\{(deliver, m_{k'_1}, \pi_{k'_1}, P_{k_1}, r, P_i, P_j), \dots, $ \par
       $(deliver, m_{k'_m}, \pi_{k'_m}, P_{k_m}, r, P_i, P_j)\}$ such that: 
            \begin{enumerate}
                \item All messages in $\mathcal{M}$ share the same triple $(r, P_i, P_j)$.
                
                \item The size of $\mathcal{M}$  is strictly larger than $\frac{n-1}{2}$. That is, $m > \frac{n-1}{2}$.
                
                \item The set of senders $S = \{P_{k_1}, \ldots, P_{k_m}\}$ are all distinct parties, i.e., the fourth element $P_{k_j}$ in each message is unique within $\mathcal{M}$.
            \end{enumerate}
            Construct $\tilde{M} = [\tilde{m}_1, \ldots, \tilde{m}_{n-1}]$, where $\tilde{m}_{k'} = m_{k'}$ \\ if $(deliver, m_{k'}, \pi_{k'}, P_{k}, r, P_i, P_j) \in \mathcal{M}$, and $\tilde{m}_{k'} = \perp$ otherwise. 
            Reconstruct the message
            $M = \mathsf{ECCdec}(\tilde{M})$. $P_j$ accepts message $M$ from $P_i$.
    \end{itemize}   
\end{algorithm}

We show that the RMT protocol for long messages satisfies the security properties of Correctness and Unforgeability under concurrent composition, assuming the underlying ECC and sequential proof scheme are secure.

\begin{theorem}
\label{thm:long-RMT}
    Assuming an $(n-1,d)$-ECC and a Sequential Proof Scheme,
    Protocol~\ref{protocol:RMT-long-concurrent} realizes RMT against any adversary $\mathcal{A}(t,c)$ under concurrent composition, provided that $n > 2c + 2t + 1$; and for any message of length $|M| \in \Omega(\lambda \log n)$, it has communication complexity $O(|M| + \lambda n \log n)$.
\end{theorem}

\begin{proof}
    \emph{Correctness:}  Since $P_i$ is an honest party, it performs Step 0 by encoding the message $M$ into $M^* = [m_1, \ldots, m_{n-1}]$, generating the proof $(\Pi, r)$, and reliably sending the tuple $(deliver, m_{k'}, \pi_{k'}, P_k, r, P_i, P_j)$ to its designated recipient $P_k$, where $k'=k$ if $k<i$ and $k'=k-1$ if $k>i$. 
    We rely on the asynchronous network's eventual delivery guarantee for messages from honest parties over unattacked channels.

    Each honest party $P_k$ that receives the initial message from $P_i$ via an unattacked channel $C_{i, k}$ will perform Step 1 and forward $(deliver, m_{k'}, \pi_{k'}, P_k, r, P_i, P_j)$ to the receiver $P_j$. This message will eventually reach $P_j$ if $C_{k, j}$ is unattacked.
    A forwarded message can be suppressed or indefinitely delayed if:
    \begin{itemize}
        \item The forwarding party $P_k$ is corrupted by the adversary ($t$ possible corruption), or
        \item The channel $C_{ik}$ (Step 0) or $C_{kj}$ (Step 1) is attacked by the adversary.
        With $c$ channel attacks available, this can suppress at most $c$ distinct forwarded messages to $P_j$.
    \end{itemize}
    The total number of distinct messages that can be suppressed is at most $t + c$.
    Since there are $n-1$ total forwarders, the number of distinct, honestly generated message that eventually reach the honest party $P_j$ is at least $n - 1 - (t + c)$.
    
    By the Completeness property of the sequential proof scheme, all symbols $m_{k'}$ and proofs $\pi_{k'}$ generated by the honest sender $P_i$ will pass the verification check $\mathsf{PVer}$ in Step 2. 
    Therefore, the set $\mathcal{T}$ of collected messages will contain at least $m = n -1 - (t + c)$ entries.
    Given the condition $n > 2c + 2t + 1$, we have:
    $$m = n - 1 - (t + c)  >  \frac{n-1}{2}.$$
    This number $m$ satisfies the condition in Step 3. 
    Specifically, $m \geq \lceil\frac{n-1}{2}\rceil \geq d$ is greater than the minimum number of correct code the ECC should have. 
    Meanwhile, by the Collision-free property of the sequential proof scheme, it is computationally infeasible for the adversary to generate a valid message that $P_i$ did not generate. 
    Thus, $\tilde{M}$ in Step 3 contains at least $d$ correct symbols, and other symbols are $\perp$. 
    By the Correctness Property of the $(n-1,d)$-erasure-correcting code, $P_j$ will successfully reconstruct the unique original message $M = \mathsf{ECCdec}(\tilde{M})$, and eventually accept $M$ from $P_i$. 

    \emph{Unforgeability:} If $P_i$ did not initiate the RMT protocol for message $M$ to $P_j$, the only way for $P_j$ to receive any message of the form $(deliver, m_{k'}, \pi_{k'}, P_k, r, P_i, P_j)$ is via adversarial intervention.
    The adversary $\mathcal{A}(t,c)$ can create and send this message through two mechanisms:
    \begin{itemize}
        \item It can corrupt up to $t$ parties. Then for each corrupted party $P_k$, it forces them to send $(deliver, m_{k'}, \pi_{k'}, P_k, r, P_i, P_j)$ to $P_j$. 
        \item Directly injecting up to $c$ messages $(deliver, m_{k'}, \pi_{k'}, P_k, r, P_i, P_j)$ over $c$ attacked channels leading to $P_j$.
    \end{itemize}
    The maximum number of distinct messages that the adversary can forge and ensure their delivery to the honest $P_j$ is at most $t + c$.
    Since $n > 2c + 2t + 1$, we have $t + c < \frac{n-1}{2}$.
    The number of forged messages $t + c$ is strictly less than the required threshold of $\frac{n-1}{2}$.
    Thus, $P_j$ will never accumulate the messages to satisfy the condition in Step 3, and therefore never accepts $M$ from $P_i$.

    \emph{Complexity:} The length of each symbols $m_{k'}$ is $\sigma \log n$ and the size of the pair $(r, \pi_{k'})$ is $O(\lambda \log n)$ when using a Merkle tree.
    The protocol involves $O(n)$ messages in total.
    Therefore, the overall complexity is the product of the message count and length: $O(n) \cdot O(\sigma \log n + \lambda \log n) = O(|M| + \lambda n \log n)$, where $\sigma = \lceil \frac{|M|}{d_{max} \log n } \rceil$ and $d_{max} = \lceil \frac{n-1}{2} \rceil$.
\end{proof}

Notice that when the message length $|M| \in \Omega(\lambda \log n)$, the communication complexity of this protocol, $O(|M| + \lambda n \log n)$, starts to be better than the complexity of Protocol~\ref{protocol:RMT-concurrent} (which is $O(n|M| + n \log n)$) as $|M|$ increases. 
This efficiency gain stems from the fact that instead of sending the entire message $M$ redundantly $O(n)$ times, only the encoded symbols and their proofs are transmitted, allowing the $O(n)$ transmissions to carry only one symbol length $\sigma \log n$ plus the proof whose length doesn't depend on $|M|$.

We classify different messages according to their lengths, and the optimal RMT protocols for different scenarios are as follows:
\begin{itemize}
    \item Short message: For messages whose length $|M| \in O(\lambda  \log n)$, we choose protocol \ref{protocol:RMT-concurrent} and the communication complexity is $O(n|M| +  n \log n)$;
    \item Long message: For messages whose length $|M| \in \Omega(\lambda \log n)$, we choose protocol \ref{protocol:RMT-long-concurrent} and the communication complexity is $O(|M| + \lambda n \log n)$.
\end{itemize}

\subsection{Reliable Broadcast for Very Long Messages}

By combining Bracha's RB protocol~\cite{bracha1987asynchronous} with our RMT Protocol~\ref{protocol:RMT-long-concurrent} for long messages, we have an RB protocol for long messages with communication complexity $O(n^2|M| + \lambda n^3 \log n)$.
In order to achieve a more efficient and concurrently composable RB protocol for longer messages, we leverage a similar idea in~\cite{das2021asynchronous} for message dissemination.

The new protocol, detailed in Protocol~\ref{protocol:broadcast-long-concurrent}, 
utilizes our RMT Protocol~\ref{protocol:RMT-long-concurrent} and RB Protocol~\ref{protocol:broadcast}, denoted by $\textsf{ConcurrentRMT}$ and $\textsf{ConcurrentRB}$ respectively when they are invoked.
%
Furthermore, Protocol~\ref{protocol:broadcast-long-concurrent} employs the Asynchronous Data Dissemination (ADD) protocol~\cite{das2021asynchronous}, which efficiently disseminates long messages in asynchronous networks: 
if at least $t+1$ honest parties are given a message $M$ as input and all honest parties are given either $M$ or $\perp$, 
then all honest parties eventually reconstruct and output $M$ (see Algorithm 1 in~\cite{das2021asynchronous}, present in Appendix~\ref{appendix:ADD} of this paper for completeness).
For concurrent composition, we replace all point-to-point communication in ADD by $\textsf{ConcurrentRMT}$ and the resulting protocol is referred as $\textsf{ConcurrentADD}$.

Protocol~\ref{protocol:broadcast-long-concurrent} is applicable to any message $M$ with $|M| \geq n \log n$, and achieves superior communication complexity compared with the construction using Protocol~\ref{protocol:RMT-long-concurrent} when the message length satisfies $|M| \in \Omega(\lambda n \log n)$. 

\begin{algorithm}[t]
\caption{Reliable Broadcast for very long messages - Concurrent Composition}
\label{protocol:broadcast-long-concurrent}
Say party $P_g$ wishes to broadcast message $M$.
The protocol proceeds as follows:
    \begin{itemize}
        \item Step 0: $P_g$ sends $(Propose, M)$ by $\mathsf{ConcurrentRMT}$ to all parties.
    \end{itemize}
    For each party $P_i$:
    \begin{itemize}
        \item Step 1: upon receiving $(Propose, M)$ from the $P_g$, let $h = hash(M)$ and act as if $P_i$ receive $(initial, h)$ from $P_g$, continue with protocol $\mathsf{ConcurrentRB}(h)$.
        \item Step 2: upon accepting $h$ from $\mathsf{ConcurrentRB}$, do:
        \begin{enumerate}
            \item if it has received $(Propose, M)$ and $h = hash(M)$ then start $\mathsf{ConcurrentADD}(M)$.
            \item else start $\mathsf{ConcurrentADD}(\perp)$.
        \end{enumerate}
    \end{itemize} 
\end{algorithm}

We now prove that Protocol~\ref{protocol:broadcast-long-concurrent} satisfies the properties of reliable broadcast under concurrent composition.

\begin{lemma}\label{lemma:RB-threshold}
    If any honest party accepts a hash $h^*$ from $\textsf{ConcurrentRB}$, then at least $t+1$ honest parties must have received a message $(Propose, M)$ from $P_g$ such that $\mathsf{hash}(M) = h^*$.
\end{lemma}
\begin{proof}
    In $\textsf{ConcurrentRB}$, for any honest party to accept $h^*$, at least one honest party must have sent a \emph{ready} message for $h^*$, which in turn requires at least $2t+1$ parties to have sent \emph{echo} messages. 
    Among these $2t+1$ parties, at least $t+1$ are honest. 
    An honest party only sends an \emph{echo} for $h^*$ if it received $h^*$ directly from the sender. 
    By the construction of Step 1, an honest party only contributes $h^*$ to the RB if it received a valid $(Propose, M)$ from $P_g$ where $\mathsf{hash}(M) = h^*$. Thus, at least $t+1$ honest parties must have received such a proposal.
\end{proof}

\begin{lemma}\label{lemma:HADD-input}
    If any honest node executes $\textsf{ConcurrentADD}(M)$ with $M \neq \perp$, then at least $t+1$ honest parties will also initiate $\textsf{ConcurrentADD}(M)$, while all other honest parties will initiate $\textsf{ConcurrentADD}(\perp)$.
\end{lemma}
\begin{proof}
    If an honest party $P_i$ executes $\textsf{ConcurrentADD}(M)$, it must have accepted $h^* = \mathsf{hash}(M)$ from $\textsf{ConcurrentRB}$. 
    By the \emph{Agreement} property of $\textsf{ConcurrentRB}$, every honest party eventually accepts the same $h^*$. 
    By Lemma~\ref{lemma:RB-threshold}, at least $t+1$ honest parties received $(Propose, M)$ from the sender. 
    These $t+1$ parties will verify $\mathsf{hash}(M) = h^*$ and thus initiate $\textsf{ConcurrentADD}(M)$. 
    Any honest party that did not receive $M$ (or received a different message $M'$ such that $\mathsf{hash}(M') \neq h^*$) will initiate $\textsf{ConcurrentADD}(\perp)$ in accordance with Step 2.2.
\end{proof}

\begin{lemma}\label{lemma:HADD-output}
    If an honest node outputs $M$, then no honest node will output $M' \neq M$, and every honest node will eventually output $M$.
\end{lemma}
\begin{proof}
    The \textsf{ConcurrentADD} primitive guarantees that if at least $t+1$ honest parties provide $M$ as input and all honest parties provide either $M$ or $\perp$, then all honest parties will eventually reconstruct and output $M$. (Lemma 3.3 in~\cite{das2021asynchronous}) 
    From Lemma~\ref{lemma:HADD-input}, we have at least $t+1$ honest parties inputting $M$ and no honest party will input a different $M' \neq M$ into $\textsf{ConcurrentADD}$. Consequently, the reconstruction property of $\textsf{ConcurrentADD}$ ensures all honest parties converge on $M$.
\end{proof}

Also, by Theorem~\ref{thm:concurrent-RMT}, we have the following lemma:
\begin{lemma}\label{lemma:HADD-concurrent}
    If $\textsf{ConcurrentADD}$ is secure in a standalone execution under $n > 3t$ faults, it remains secure under concurrent composition with $n > 2c + 2t + 1$ and $n > 3t$.
\end{lemma}

Finally, we have the following:
\begin{theorem}
\label{thm:long-RB}
    Protocol~\ref{protocol:broadcast-long-concurrent}
    realizes RB under concurrent composition against any adversary $\mathcal{A}(t,c)$, provided that $n > 2c + 2t + 1$ and $n > 3t$; and for message length $|M| \in \Omega(\lambda n \log n)$, it has communication complexity $O(n|M| + \lambda n^3 \log n)$.
\end{theorem}
\begin{proof}
    \emph{Validity:} If the sender $P_g$ is honest and broadcasts $M$, then by the \emph{Correctness} of \textsf{ConcurrentRMT}, all honest parties receive $(Propose, M)$. 
    Consequently, all honest parties initiate \textsf{ConcurrentRB} with $h = \mathsf{hash}(M)$. By the \emph{Validity} of \textsf{ConcurrentRB}, all honest parties accept $h$. 
    Since they all possess $M$, they all initiate \textsf{ConcurrentADD}(M). By the properties of \textsf{ConcurrentADD} with $n > 3t$, all honest parties will eventually output $M$.
    
    \emph{Agreement:} Suppose an honest party $P_i$ outputs $M$. 
    This implies $P_i$ accepted $h = \mathsf{hash}(M)$ from \textsf{ConcurrentRB}.
    By the \emph{Agreement} property of \textsf{ConcurrentRB}, any other honest party $P_j$ must have accepted the same $h$. 
    As shown in Lemma~\ref{lemma:HADD-output}, the collision resistance of the hash function ensures that no honest party inputs a different $M' \neq M$ into \textsf{ConcurrentADD}. 
    The security of \textsf{ConcurrentADD} under concurrent composition (Lemma~\ref{lemma:HADD-concurrent}) then ensures that all honest parties output the same message $M$.

    \emph{Complexity:} 
    The original ADD protocol involves $n^2$ point-to-point communications.
    For messages of length $|M| \in \Omega(\lambda n \log n)$, each communication in the original ADD protocol satisfies the length condition required for \textsf{ConcurrentRMT}.
    Consequently, applying our \textsf{ConcurrentRMT} compiler to ADD transforms its communication complexity from $O(n|M| + \lambda n^2)$ to $O(n|M| + \lambda n^3 \log n)$ for \textsf{ConcurrentADD}.
    
    The communication complexity of Protocol~\ref{protocol:broadcast-long-concurrent} consists of three parts: 
    \begin{itemize}
        \item the $n$ invocations of $\textsf{ConcurrentRMT}$ in Step 0 requires $n \cdot O(|M| + \lambda n \log n)$;

        \item one global invocation of 
    $\textsf{ConcurrentRB}$ for a hash requires $O(n^3 \lambda + n^3 \log n)$; and 

        \item for $|M| \in \Omega(\lambda n \log n)$, one global invocation of the $\textsf{ConcurrentADD}$ protocol requires  $O(n|M| + \lambda n^3 \log n)$.
    \end{itemize}
Adding them together, therefore, the total communication complexity of the protocol is $O(n|M| + \lambda n^3 \log n)$. 
\end{proof}

We now further expand the selection of the optimal RB protocol under different message lengths as follows:
\begin{itemize}
    \item Short message: For messages of length $|M| \in O(\lambda \log n)$, we use Protocol~\ref{protocol:broadcast} and the communication complexity is $O(n^3|M| +  n^3 \log n)$;
    \item Long message: For messages of length $|M| \in \Omega(\lambda  \log n) \cap O(\lambda n \log n)$, we use Bracha's RB protocol compiled with our RMT Protocol~\ref{protocol:RMT-long-concurrent}, and the communication complexity is $O(n^2|M| + \lambda n^3 \log n)$;
    \item Very long message: For messages of length $|M| \in \Omega(\lambda n \log n)$, we use Protocol \ref{protocol:broadcast-long-concurrent} and the communication complexity is $O(n|M| + \lambda n^3 \log n)$.
\end{itemize}

\subsection{Complexity Analysis of Compiled Protocols}

Table~\ref{table:before-after} in Section~\ref{sec:intro} summarizes the communication complexity of several representative asynchronous consensus protocols before and after applying our compilers.
The compiled complexity depends on the original protocol's message length and the appropriate choice of our primitives (short, long, or very long message versions).
Indeed, even within the same protocol, different operations may have different message lengths and correspond to different primitives in our transformation. 
Noticeably, the primitives' communication complexity typically has a message-dependent component and an additive message-independent component. Thus the compiled complexity is not calculated by direct multiplication of the primitive's complexity with the number of times when they are applied. Instead, it can be upper-bounded by the sum of the message-dependent costs across all invocations, where the total message length is the original protocol's communication complexity, and the cumulative message-independent cost of the primitives, which is the number of messages times the message-independent component.

We now elaborate on how our compiled protocol is derived for each original protocol, as well as the resulting communication complexity.

\paragraph{\textbf{Binary BA Protocol}}
The protocol of Huang et al.~\cite{huang2024byzantine} is the first binary asynchronous BA protocol that achieves polynomial communication complexity without requiring any cryptographic assumption.
Since the protocol operates on messages of single bits, we directly replace all point-to-point message transmissions with our short message RMT primitive, which has complexity $O(n|M| + n\log n)$.
When the number of Byzantine parties satisfies $n = 3t + 1$, the original protocol achieves communication complexity $\tilde{O}(n^{12})$, where $\tilde{O}$ omits poly-logarithmic factors in $n$.
Therefore, the complexity of the compiled protocol is $\tilde{O}(n^{13})$ for binary BA, with an extra $\log n$ factor absorbed by the $\tilde{O}$ notation.
For multi-bit inputs, the compiled protocol directly executes multiple instances concurrently, yielding compiled complexity $\tilde{O}(n^{13} |L|)$, where $|L|$ is the input length.

\paragraph{\textbf{Asynchronous Common Subset using Hash}}
The protocol of Das et al.~\cite{das2024asynchronous} achieves Asynchronous Common Subset (ACS)
with communication complexity of $O(\lambda n^3)$ for short inputs, 
relying only on cryptographic hash functions modeled as a random oracle.
The protocol consists of a standard ACS that calls $n$ number of reliable broadcasts, and an Index ACS, which further calls Index Validated Asynchronous Byzantine Agreement (VABA) and other subprotocols, including Asynchronous Secret Key Sharing (ASKS), Index Gather, and Reliable Agreement.
\begin{itemize}
    \item For the case where $|L| \in O(\lambda \log n)$, we directly replace all point-to-point message transmissions with our short message RMT primitives, which have complexity $O(n|M| + n \log n)$.
    As the original protocol achieves a message complexity of $O(n^3)$ and a total communication complexity of $O(\lambda n^3)$,
    after compilation the concurrently secure protocol has message-dependent complexity component $O(\lambda n^3) \cdot O(n) = O(\lambda n^4)$ and message-independent complexity component $O(n^3) \cdot O(n \log n) = O(n^4 \log n)$.
    Therefore, assuming $\lambda > \log n$ as common practice, the compiled protocol achieves total communication complexity $O(\lambda n^4 + n^4 \log n) = O(\lambda n^4)$.
    
    \item For the case where $|L| \in O(\lambda n \log n)$, in the original protocol the $n$ reliable broadcasts actually have a total communication complexity of $O(\lambda n^3 \log n)$, which dominates the $O(\lambda n^3)$ communication of the Index ACS component.
    To achieve concurrent security,  we use our RB compiler for very long message, with complexity $O(n|M| + \lambda n^3 \log n)$,
    thus the reliable broadcast yields communication complexity $O(n) \cdot O(\lambda n^3 \log n) = O(\lambda n^4 \log n)$.
    At the same time, the Index ACS component contributes $O(\lambda n^4)$ communication complexity for the compiled protocol, same as in the previous cases.
    Therefore, the total communication complexity of the compiled protocol is $O(\lambda n^4 \log n)$.
    
    \item For the case where $|L| \in \Omega(\lambda n^2 \log n)$, the original protocol has complexity $O(n^2|L|)$, as the reliable broadcast of long input messages dominates all other components. To achieve concurrent security, we also use our RB primitive for very long messages, and the additive overhead $O(\lambda n^3 \log n)$ is negligible compared to $O(n^2|L|)$ when $|L| \in \Omega(\lambda n^2 \log n)$.
    Therefore, the total communication complexity of the compiled protocol remains $O(n^2|L|)$.
\end{itemize}

\paragraph{\textbf{Asynchronous Common Subset using Public-key Infrastructure}}
The protocol of Guo et al.~\cite{guo2022speeding} achieves Asynchronous Common Subset based on threshold public key encryption. It has 
communication complexity $O(n^2|L|+\lambda n^3\log n)$, achieving optimal complexity for the broadcast phase for long inputs $|L| \in \Omega(\lambda n \log n)$, and optimal message complexity of $O(n^2)$.
The protocol structure consists of three main phases: a broadcast phase using the $n$ parallel PB (Provable Broadcast) components, each with communication complexity $O(n|L|+\lambda n)$, an MVBA phase using Speeding MVBA (sMVBA) protocol with complexity $O(\lambda n^2)$, and a recovery phase for handling weak consistency of PB that contributes $O(n^2|L|+\lambda n^3\log n)$ complexity.
\begin{itemize}
    \item For the case where $|L| \in O(\lambda n \log n)$, the original protocol achieves communication complexity of $O(\lambda n^3 \log n)$.
    For the compiled protocol, 
    we use our short-message RMT primitive for sMVBA, very-long-message RB for PB, and very-long-message RMT for messages in the recovery phase.
    
    By replacing PB with the very-long-message RB protocol, the complexity of the broadcast phase becomes $O(n) \cdot O(\lambda n^3 \log n) =O(\lambda n^4 \log n)$.
    The sMVBA phase, compiled with the short-message RMT protocol, contributes at most $O(\lambda n^2) \cdot O(n \log n) = O(\lambda n^3 \log n)$ complexity.
    Meanwhile, the recovery phase has message length $O(|L|+ \lambda \log n)$ and at most $O(n^2)$ messages, so after compiling with the very-long-message RMT protocol, it contributes a complexity of $O(\lambda n^3 \log n)$.
    Therefore, the total communication complexity of the compiled protocol is $O(\lambda n^4 \log n)$.

    \item For the case where $|L| \in \Omega(\lambda n^2 \log n)$, the original protocol achieves complexity $O(n^2|L|)$, as this term dominates the $\lambda n^3 \log n$ term when $|L| \in \Omega(\lambda n^2 \log n)$.
    To achieve concurrent security, the usage of our primitives is the same as above. Most of the analysis remains the same, except that the recovery phase now contributes a complexity of $O(n^2|L|)$.
    Thus the total communication complexity of the compiled protocol is $O(n^2|L|)$.
\end{itemize}
\section{Conclusion and Future Directions}
The primary contributions of this work include the introduction of a novel adversarial model that simultaneously enables the corruption of parties and communication channels. 
The reorder attack is both theoretically interesting and practically viable, 
enhancing the security analysis for Byzantine Agreement protocols under both parallel and concurrent executions.

Building upon this model, our work establishes strong positive and negative results. On the one hand, authenticated Byzantine Agreement becomes insecure when either $n \leq 3t$ or $n \leq 2c + 2t + 1$. 
While on the other hand, for unauthenticated Byzantine Agreement, secure protocols exist under both parallel and concurrent compositions when $n > \max\{3t, 2c+2t+1\}$. 
Notably, these findings provide tight conditions for the security of Byzantine Agreement under compositional executions ---a fundamental advance in understanding protocol resilience against corruption and reorder attacks.

The framework and results established in this paper motivate several promising directions for future research:
\begin{itemize}
    \item In terms of positive results, we provide black-box compilers for both parallel and concurrent executions of BA protocols. 
    For sufficiently long messages, our optimized RB and RMT protocols achieve communication complexities of $O(n|M|)$ and $O(|M|)$, respectively, modulo an additive term that doesn't rely on $|M|$. Ignoring the additive term, they match the asymptotic optimality of standard (non-compositional) RB and RMT protocols. 
    This suggests that further improvements in communication complexity would require fundamentally different technical insights beyond current approaches, which would be both interesting and challenging to investigate. 

    \item
    Another important open question is to establish tight lower bounds on the communication complexity of compositionally-secure BA protocols, which would provide a theoretical foundation for understanding the inherent cost of security under reorder attacks.
    Intuitively, such attacks necessitate additional communication to preserve security, but the exact complexity trade-offs require further studies.


    \item Finally, in this work we focus on the standard notion of validity commonly considered in the literature for BA protocols. Since other validity notions such as {\em weak validity} have also been studied \cite{lamport1983weak}, we are interested in whether our conclusions still hold with respect to those notions. Intuitively, weaker validity conditions would relax the constraints that BA protocols must satisfy, hence may allow the impossibility results to be circumvented. Thus extending our current model to weaker validity notions constitutes a promising direction for the future.
\end{itemize}

\bibliographystyle{IEEEtran}
\bibliography{ref}

\bigskip
\appendix
\section*{Consensus Protocol in Figure 4 of~\cite{bracha1987asynchronous}}\label{appendix:bracha-BA}
This section presents the unauthenticated BA protocol from Bracha 1987~\cite{bracha1987asynchronous}, which serves as a building block in our black-box compiler for concurrent composition.
The protocol operates in phases, where each phase consists of three steps of reliable broadcast operations, and uses coin tossing to ensure termination.

\begin{algorithm}[ht]
  \caption{The Consensus Protocol in Figure 4 of~\cite{bracha1987asynchronous}}
  \label{protol:consensus}
  
  \textbf{Phase}(i): (by process $p$)\;
  
  \textbf{Step 1:} $Broadcast (p, 3i+1, value_p)$. Wait until validate $n-t$ $(3i+1)$-messages. 
    \begin{itemize}
        \item $value_p :=$ majority value of the $n-t$ validated messages.
    \end{itemize}
  \textbf{Step 2:} $Broadcast(p, 3i+2, value_p)$. Wait until validate $n-t$ $(3i+2)$-messages.
    \begin{itemize}
      \item \textbf{(i)} If more than $\frac{n}{2}$ of the messages have the same value $v$, then $value_p = (d,v)$.
      \item \textbf{(ii)} Otherwise, $value_p := value_p$.
    \end{itemize}
  \textbf{Step 3:} $Broadcast(p, 3i+3, value_p)$. Wait until validate $n-t$ $(3i+3)$-messages.
    \begin{itemize}
      \item \textbf{(i)} If validated more than $2t$ messages with value $(d,v)$ then $decision_p := value_p := v$.
      \item \textbf{(ii)} If validated more than $t$ messages with value $(d,v)$ then $value_p := v$.
      \item \textbf{(iii)}  Otherwise, $value_p := coin\_toss$ (0 or 1 with probability $\frac{1}{2}$).
    \end{itemize}
Go to round 1 of phase $i+1$ \;
\end{algorithm}

\section*{Asynchronous Data Dissemination}\label{appendix:ADD}
This section presents the Asynchronous Data Dissemination (ADD) protocol from Das et al.~\cite{das2021asynchronous}, which efficiently disseminates long messages in asynchronous networks.
The protocol uses Reed-Solomon error-correcting codes to encode messages into shares, and ensures that if at least $t+1$ honest parties provide the same message~$M$ as input, all honest parties eventually reconstruct and output $M$.

\begin{algorithm}[ht]
  \caption{Asynchronous Data Dissemination (ADD) of~\cite{das2021asynchronous}}
  \label{protol:ADD}

  \textbf{input $M_i$}: either $M_i = M$ or $M_i = \perp$
  
  if $M_i \neq \perp$ then
  \begin{itemize}
      \item Let $M' := [m_1, m_2, \ldots, m_{n}] := \mathsf{RSEnc}(M_i, n, t+1)$
  \end{itemize}
  
  if $M_i \neq \perp$ then
  \begin{itemize}
      \item Let $m^*_i := m_i$
      \item send $<DISPERSE, m_j>$ to node $j$ for every $j = 1,2,\ldots,n$
  \end{itemize}
  else
  \begin{itemize}
      \item upon receiving $t+1$ identical $<DISPERSE, m_i>$ do
      \item Let $m^*_i := m_i$.
  \end{itemize}
  
  // reconstruction phase

  send $<RECONSTRUCT, m^*_i>$ to all nodes \\
  if $M_i \neq \perp$ then
  \begin{itemize}
      \item output $M$ and return
  \end{itemize}

  Let $T := \{\}$ \\
  For every $<RECONSTRUCT, m^*_j>$ received from node $j$, add $(j, m^*_j)$ to $T$.
  For $0 \leq r \leq t$ do:
  \begin{itemize}
      \item Wait till $|T| \geq 2t+r+1$
      \item Let $p_r(\cdot) := \mathsf{RSDec}(t+1,r,T)$
      \item If $2t+1$ elements $(j,a) \in T$ satisfy $p_r(j)= a$ then output coefficient of $p_r(\cdot)$ as $M$ and return
  \end{itemize}

\end{algorithm}

\section*{Proof of Theorem \ref{thm:impossibility}}\label{appendix:proof of impossibility}

\begin{proof}

    For condition (1), the result follows directly from Theorem 1 in~\cite{lindell2006composition}.

    We prove the result for condition (2) by contradiction.
    Assume that there exists an authenticated Byzantine Agreement protocol $\Pi$ that remains secure under two parallel executions against an adversary $\mathcal{A}(t,c)$ with $n \leq 2c+2t+1$. We assume $c > 0$, otherwise we have $2t+1 \geq n$ and it can be reduced to condition (1).

    We consider two independent executions of $\Pi$: Let $A_1, A_2, \cdots, A_n$ and $B_1, \cdots, B_n$ be independent copies of $n$ parties participating in protocol $\Pi$. The independent copies mean that for each $i$, $A_i$ and $B_i$ are the same party that runs in two different parallel executions of $\Pi$. For the purpose of the proof below, we denote the set $T = \{A_2,\cdots,A_{\lfloor \frac{n+1}{2} \rfloor}\}$, $\widetilde{T} = \{B_2,\cdots,B_{\lfloor \frac{n+1}{2} \rfloor}\}$, $H = \{A_{\lfloor \frac{n+3}{2} \rfloor},\cdots,A_{n}\}$, 
    and $\widetilde{H} = \{B_{\lfloor \frac{n+3}{2} \rfloor},\cdots,B_{n}\}$. Note that we have $n \leq 2c+2t+1$, so we have $c + t \geq |T| = |\widetilde{T}|$ and $c + t \geq |H| = |\widetilde{H}|$. 
    
    Let $rounds(\Pi)$ denote the maximum number of communication rounds required for protocol $\Pi$ to achieve termination in any execution. By the termination guarantees of the agreement property of Byzantine Agreement protocols, $rounds(\Pi)$ is finite and well-defined. Furthermore, given the security assumption that $\Pi$ remains secure under parallel composition, it follows that $\Pi$ must terminate within $rounds(\Pi)$ rounds even when executed concurrently across multiple protocol instances.

    We now introduce an important abstract concept and will instantiate it in different ways in the proof. A system \textsf{X} is a tuple:
    \begin{itemize}
         \item A set $P=\{P_1,\dots,P_n\}$ of $n$ parties, each party $P_j$ runs two copies of ITMs, $A_j$ and $B_j$, for protocol $\Pi$; 
         \item An adversary $\mathcal{A}(t,c)$ controlling $t$ corrupted parties and reordering $c$ channels;
         \item Initial input values for all ITMs on their input tapes;
         \item A network topology governing inter-party connectivity.
    \end{itemize}

Intuitively, if the system contains a non-trivial adversary who corrupts some parties and reorders some channels, then the network topology is consistent with the adversary's reordering of the channels.
In some mental games, we may instantiate the system without an adversary, in which case the network topology dictates how the parties' communication channels are inter-connected crossing the two executions of $\Pi$.
    
    For any ITM $M$ run by a party in $P$, its \emph{view} in system \textsf{X}, denoted by $\mathsf{view}_X(M)$, consists of
the content of $M$'s input tapes and random tape, where the former in particular includes the initial input value and the messages received by $M$ from other parties during the execution of \textsf{X}.

    If all parties (including $\mathcal{A}(t,c)$) execute deterministically, then $ \mathsf{view}_X(M) $ is completely specified by the initial tape contents and the actions of the adversary. In contrast, if randomness is involved (with random tapes drawn uniformly at random), then $ \mathsf{view}_X(M) $ is a random variable defined on the corresponding probability space.

    \begin{figure}
    \centering
    \includegraphics[width=0.9\linewidth]{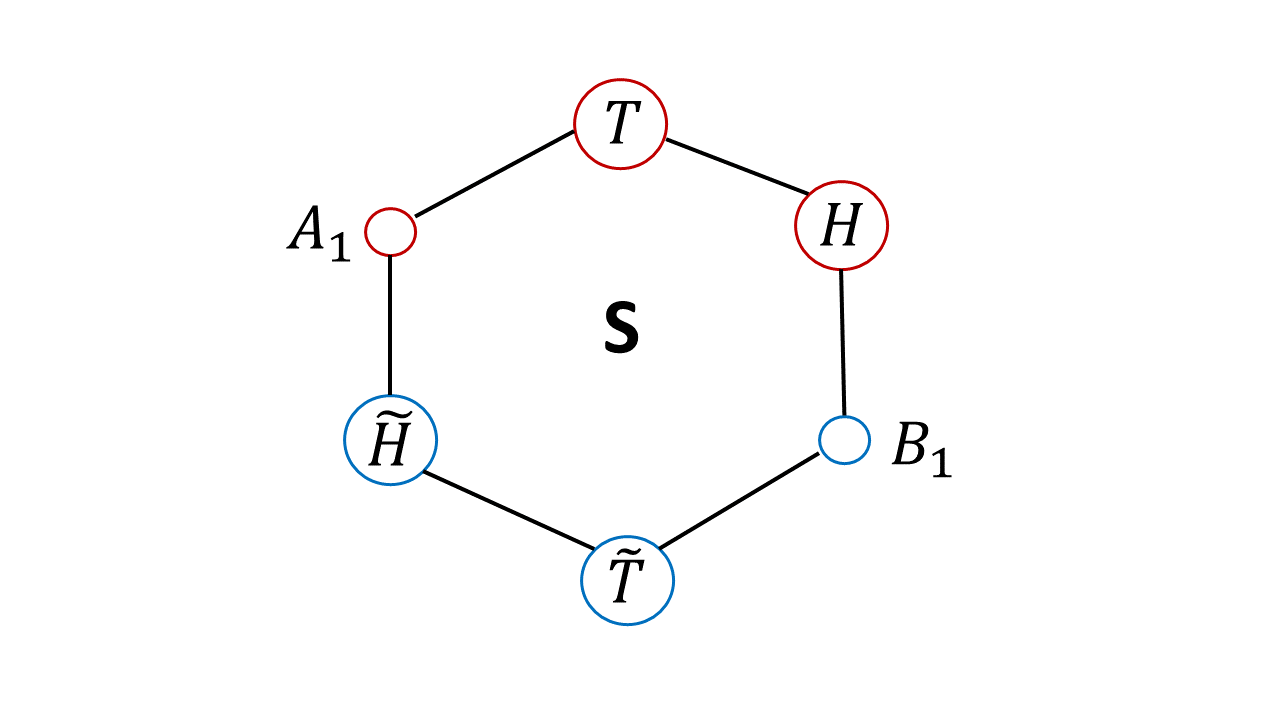}
    \caption{Network topology of System \textsf{S}.}
    \label{fig:system s}
\end{figure}

    We now instantiate the system \textsf{S} through the following configuration:
\begin{itemize}
    \item \textbf{Network Topology}: ITM $ A_1 $ establishes connectivity with the set $ \widetilde{H} $ (replacing its original link to $ H $), while ITM $ B_1 $ interfaces with $ H $ instead of $ \widetilde{H} $. This structural reconfiguration is illustrated in Figure~\ref{fig:system s}.
    
    \item \textbf{Initial Input Value}: 
    \begin{itemize}
        \item ITMs in $ \{A_1\} \cup T \cup H $ receive input $ 0 $;
        \item ITMs in $ \{B_1\} \cup \widetilde{T} \cup \widetilde{H} $ receive input $ 1 $.
    \end{itemize}

    \item \textbf{Adversary}: In this system, there is no adversary; formally, we have $t = 0$ and $c = 0$.

    \item \textbf{Protocol Execution}: All ITMs in $\mathsf{S}$ strictly adhere to the instruction of protocol $ \Pi $, simulating an \emph{authenticated Byzantine Agreement} environment. Specifically, each party:
    \begin{itemize}
        \item Generates messages as prescribed by $ \Pi $’s honest execution semantics;
        \item Sends these messages via the  network topology defined above.
    \end{itemize}
\end{itemize}

    Now we first state the following lemmas and we will prove them afterwards.
    \begin{lemma}
        In system \textsf{S}, all ITMs in $\{A_1\} \cup T$ output 0, and there exists at least one ITM in $H$ outputs 0. Symmetrically, all ITMs in $\{B_1\} \cup \widetilde{T}$ output 1, and there exists at least one ITM in $\widetilde{H}$ outputs 1.
        \label{lemma:system 2PBA}
    \end{lemma}

\begin{lemma}
    In system \textsf{S}, all ITMs in $\{A_1\} \cup \widetilde{H}$ and at least one ITM in $\widetilde{T}$ output the same value; Symmetrically, all ITMs in $\{B_1\} \cup H$ and at least one ITM in $T$ outputs the same value.
\label{lemma:system 2PBA'}
\end{lemma}

Combining Lemma~\ref{lemma:system 2PBA} with Lemma~\ref{lemma:system 2PBA'} we have $Output_S(\widetilde{H}) = Output_S(A_1) = 0$ and at least one ITM $\widetilde{h} \in \widetilde{H}$ satisfies $Output_S(\widetilde{h}) = Output_S(B_1) = 1$. 
Hence, we derive a contradiction on the output of $\widetilde{h}$, thereby completing the proof of the theorem. 
\end{proof}

Next, we complete the proofs of Lemma \ref{lemma:system 2PBA} and Lemma \ref{lemma:system 2PBA'}.

\begin{proof}[Proof of Lemma \ref{lemma:system 2PBA}]

To prove the lemma, we define a new system \textsf{2PBA} as follows and as shown in Figure~\ref{fig:system 2pba}:

\begin{figure}
    \centering
    \includegraphics[width=0.9\linewidth]{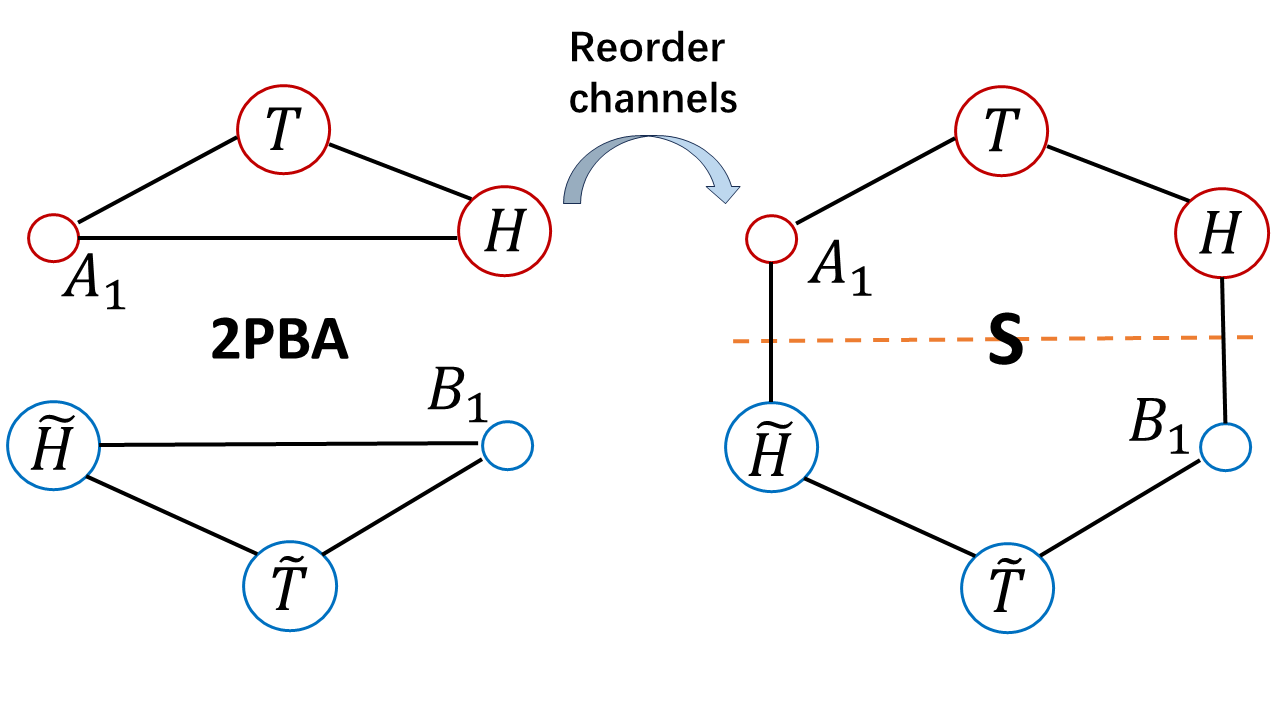}
    \caption{A party's view is identical between System \textsf{2PBA} and System \textsf{S}.}
    \label{fig:system 2pba}
\end{figure}

\begin{itemize}
    \item \textbf{Network Topology}: The system consists of two parallel executions of protocol $\Pi$, denoted $\Pi_0$ (consisted of $A_1$, $T$ and $H$) and $\Pi_1$ (consisted of $B_1$, $\widetilde{T}$ and $\widetilde{H}$). 
    
    \item \textbf{Initial Input Value}:
    \begin{itemize}
        \item ITMs in $ \{A_1\} \cup T \cup H $ receive input $ 0 $;
        \item ITMs in $ \{B_1\} \cup \widetilde{T} \cup \widetilde{H} $ receive input $ 1 $.
    \end{itemize}

    \item \textbf{Adversary}: The adversary $\mathcal{A}(t,c)$ chooses an arbitrary subset of parties $\{P_{\lfloor \frac{n+3}{2} \rfloor},\cdots,P_n\}$, denoted by $C$, which satisfies $|C| = \min\{c, n - \lfloor \frac{n+3}{2} \rfloor + 1\}$. 
    Then the adversary reorders all the channels between $C$ and $P_1$, and corrupts the parties in $\{P_{\lfloor \frac{n+3}{2} \rfloor},\cdots,P_n\} \setminus C$. As we have $2c+2t+1 \geq n$, the adversary can indeed corrupt all these parties. In the following proof, we denote $ITM^C_0 = \{ITM_0^j : j \in C\}$ as the set of ITMs run by $C$ in the protocol $\Pi_0$, and as the same, we denote $ITM^C_1 = \{ITM_1^j : j \in C\}$. 
    
    \item \textbf{Protocol Execution}:
    \begin{itemize}
        \item All honest parties in $\mathsf{S}$ strictly adhere to the instruction of protocol $ \Pi $;
        \item The corrupted parties reorder their channels between $P_1$, and besides this action, they adhere to the instruction of protocol $ \Pi $.
    \end{itemize}
\end{itemize}

We then prove the following view equivalences:
$\mathsf{view}_{2PBA}(A_1) = \mathsf{view}_S(A_1)$,
$\mathsf{view}_{2PBA}(T) = \mathsf{view}_S(T)$,
and $\mathsf{view}_{2PBA}(H) = \mathsf{view}_S(\widetilde{H})$.

We proceed by induction on the round counter $j \geq 1$.

\noindent \textbf{Base Case ($j=1$)}:
\begin{itemize}
    \item \emph{Message Generation}: All parties in both $\mathsf{2PBA}$ and $\mathsf{S}$ generate identical message sets due to protocol $\Pi$ strictly.
    
    \item \emph{Message Routing}:
    \begin{itemize}
        \item In $\mathsf{S}$: Messages propagate through the hexagonal graph.
        \item In $\mathsf{2PBA}$: The adversary $\mathcal{A}(t,c)$ enforces the following redirections:
        $$
        \begin{aligned}
            A_1 \to H \stackrel{\text{reorder}}{\mapsto} \widetilde{H}&, \quad 
            \widetilde{H} \to B_1 \stackrel{\text{reorder}}{\mapsto} A_1, \\
            A_1 \to T & , \quad T \to H \\
            H \to T & , \quad T \to A_1 \\
            H \to A_1 \stackrel{\text{reorder}}{\mapsto} B_1&, \quad 
            B_1 \to \widetilde{H} \stackrel{\text{reorder}}{\mapsto} H.
        \end{aligned}
        $$
    \end{itemize}
    After reordering, $B_1$ communicates with $H$ rather than $\widetilde{H}$, $A_1$ communicates with $\widetilde{H}$ rather than $H$. Other than that, all the ITMs strictly adhere to the instruction of $\Pi$. It is easy to see that the actual message routing is the hexagonal graph the same as in \textsf{S}.
    \smallskip
    
    \item \emph{View Equivalence}: The above routing rules induce identical message distributions at all parties' interfaces. Thus, we have
    $\mathsf{view}_{2PBA}(A_1) = \mathsf{view}_S(A_1)$,
    $\mathsf{view}_{2PBA}(T) = \mathsf{view}_S(T)$,
    and $\mathsf{view}_{2PBA}(H) = \mathsf{view}_S(\widetilde{H})$.
\end{itemize}

\noindent \textbf{Inductive Step}: Assume equivalence holds through step $j$. For step $j+1$:
\begin{itemize}
    \item \emph{Message Generation}: Identical message distributions emerge from equivalent historical views (by inductive hypothesis).
    
    \item \emph{Message Routing}: The hexagonal routing in $S$ and adversarial redirection in $\mathsf{2PBA}$ preserve the same correspondence as $j=1$.
    
    \item \emph{View Update}: Therefore, $\mathsf{view}^{(j+1)}_{2PBA}(A_1) = \mathsf{view}^{(j+1)}_S(A_1)$ and $\mathsf{view}^{(j+1)}_{2PBA}(T) = \mathsf{view}^{(j+1)}_S(T)$ maintain for the $j+1$ rounds.
\end{itemize}

Owing to the termination property of Byzantine Agreement, which mandates that the protocol $\Pi$ halts within $rounds(\Pi)$, the aforementioned procedure completes in at most $rounds(\Pi)$ rounds. Consequently, we derive that $\mathsf{view}_{2PBA}(A_1) = \mathsf{view}_S(A_1)$, $\mathsf{view}_{2PBA}(T) = \mathsf{view}_S(T)$, and $\mathsf{view}_{2PBA}(H) = \mathsf{view}_S(\widetilde{H})$.

    Given the equivalence of views, the parties $A_1$, $T$ and $H$ in both systems \textsf{2PBA} and~\textsf{S} must yield identical final outputs. By the validity property of Byzantine Agreement, since 
    $A_1$, $T$, and $H$ in $\Pi_0$ are initialized with input 0 and $A_1$, $T$, $ITM^C_0 \subset H$ are honest, then $Output_{2PBA}(A_1) = Output_{2PBA}(T) = Output_{2PBA}(ITM^C_0)= 0$. Consequently, system \textsf{S} inherits this outcome, we have $Output_{S}(A_1) = Output_{S}(T) = Output_{S}(ITM^C_0) = 0$.

Analogously, through the system \textsf{2PBA}, we can rigorously demonstrate the following view equivalences:
$\mathsf{view}_{2PBA}(B_1) = \mathsf{view}_S(B_1)$,
$\mathsf{view}_{2PBA}(\widetilde{T}) = \mathsf{view}_S(\widetilde{T})$,
and $\mathsf{view}_{2PBA}(\widetilde{H}) = \mathsf{view}_S(\widetilde{H})$,
thereby establishing $Output_{S}(B_1) = Output_{S}(\widetilde{T}) = Output_{S}(ITM^C_1) = 1$.
\end{proof}

\begin{proof}[Proof of Lemma \ref{lemma:system 2PBA'}]
Similarly, we construct  system \textsf{2PBA'} as instantiated below, shown in Figure~\ref{fig:system 2pba'}:

\begin{figure}
    \centering
    \includegraphics[width=0.9\linewidth]{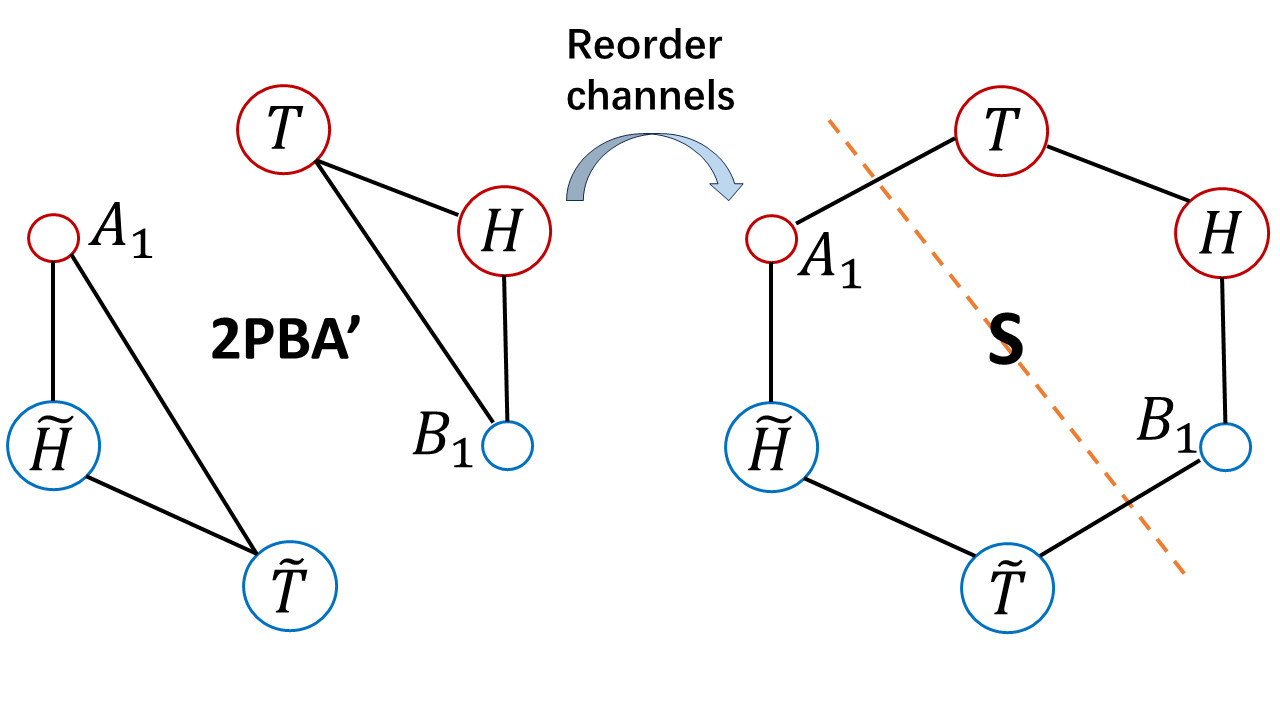}
    \caption{A party's view is identical between System \textsf{2PBA'} and System \textsf{S}.}
    \label{fig:system 2pba'}
\end{figure}
    \begin{itemize}
    \item \textbf{Network Topology}: The system consists of two parallel executions of protocol $\Pi$, denoted $\Pi'_0$ (consisted of $A_1$, $\widetilde{H}$ and $\widetilde{T}$) and $\Pi'_1$ (consisted of $B_1$, $H$ and $T$). 
    
    \item \textbf{Initial Input Value}:
    \begin{itemize}
        \item ITM $ \{A_1\}$ receives input $ 0 $ and ITMs in $\widetilde{H} \cup \widetilde{T}$ receive input $ 1 $;
        \item ITM $ \{B_1\}$ receives input $ 0 $ and ITMs in $H \cup T$ receive input $ 1 $.
    \end{itemize}

    \item \textbf{Adversary}: 
    The adversary $\mathcal{A}(t,c)$ chooses an arbitrary subset of parties $\{P_2,\cdots,P_{\lfloor \frac{n+1}{2} \rfloor}\}$, denoted by $C'$, which satisfies $|C'| = \min\{c, \lfloor \frac{n+1}{2} \rfloor - 1\}$. Then the adversary reorders all the channels between $C'$ and $P_1$, and corrupts the parties in $\{P_2,\cdots,P_{\lfloor \frac{n+1}{2} \rfloor}\} \setminus C'$. As we have $2c+2t+1 \geq n$, the adversary can indeed corrupt all these parties. In the following proof, we denote $ITM^{C'}_0 = \{ITM_0^j : j \in C'\}$ as the set of ITMs run by $C'$ in the protocol $\Pi_0$, and as the same, we denote $ITM^{C'}_1 = \{ITM_1^j : j \in C'\}$.
    
    \item \textbf{Protocol Execution}:
    \begin{itemize}
        \item All honest parties in $\mathsf{S}$ strictly adhere to the instruction of protocol $\Pi$;
        \item The corrupted parties reorder their channels between $P_1$, and besides this action, they adhere to the instruction of protocol $ \Pi $.
    \end{itemize}
\end{itemize}

The subsequent proof by mathematical induction is analogous to the reasoning in system \textsf{2PBA} of the preceding proof; therefore it is omitted here. 
It is noteworthy that after we get $\mathsf{view}_{2PBA'}(A_1) = \mathsf{view}_S(A_1)$, $\mathsf{view}_{2PBA'}(\widetilde{H}) = \mathsf{view}_S(\widetilde{H})$, and $\mathsf{view}_{2PBA'}(\widetilde{T}) = \mathsf{view}_S(\widetilde{T})$, 
we use the Agreement property (rather than Validity) of Byzantine agreement to deduce $Output_{S}(A_1) = Output_{S}(\widetilde{H}) = Output_{S}(ITM^{C'}_1)$. 
Similarly, $Output_{S}(B_1) = Output_{S}(H) = Output_{S}(ITM^{C'}_0)$.
\end{proof}

\end{document}